\documentclass{article}

\usepackage[final]{neurips_2023}
\setcitestyle{square,numbers,comma}

\usepackage[utf8]{inputenc} 
\usepackage[T1]{fontenc}    
\usepackage{hyperref}       
\usepackage{url}            
\usepackage{booktabs}       
\usepackage{amsfonts}       
\usepackage{nicefrac}       
\usepackage{microtype}      
\usepackage{xcolor}         

\usepackage{colortbl}
\usepackage{graphicx}
\usepackage{multirow}

\usepackage{amsmath}
\usepackage{amssymb}
\usepackage{mathtools}
\usepackage{amsthm}
\usepackage{comment}
\usepackage{wrapfig}
\usepackage{soul}

\usepackage[capitalize,noabbrev]{cleveref}
\usepackage{cancel}
\usepackage{changepage}

\newtheorem{theorem}{Theorem}[section]
\newtheorem{proposition}[theorem]{Proposition}

\theoremstyle{definition}

\theoremstyle{remark}

\title{GeoTMI: \\Predicting Quantum Chemical Property with Easy-to-Obtain Geometry via Positional Denoising}

\author{
  Hyeonsu Kim\thanks{Equal contributors.}\\ 
  Department of Chemistry\\
  KAIST\\
  Daejeon, South Korea \\
  \And
  Jeheon Woo$^*$\\ 
  Department of Chemistry\\
  KAIST\\
  Daejeon, South Korea \\
  \And
  Seonghwan Kim$^*$\\ 
  Department of Chemistry\\
  KAIST\\
  Daejeon, South Korea \\
  \And
  Seokhyun Moon$^*$\\ 
  Department of Chemistry\\
  KAIST\\
  Daejeon, South Korea \\
  \And
  Jun Hyeong Kim$^*$\\ 
  Department of Chemistry\\
  KAIST\\
  Daejeon, South Korea \\
  \And
  Woo Youn Kim\thanks{Corresponding author: wooyoun@kaist.ac.kr}\\ 
  Department of Chemistry\\
  KAIST\\
  Daejeon, South Korea \\
}

\begin{document}

\maketitle

\begin{abstract}
As quantum chemical properties have a dependence on their geometries, graph neural networks (GNNs) using 3D geometric information have achieved high prediction accuracy in many tasks. However, they often require 3D geometries obtained from high-level quantum mechanical calculations, which are practically infeasible, limiting their applicability to real-world problems. To tackle this, we propose a new training framework, GeoTMI, that employs denoising process to predict properties accurately using easy-to-obtain geometries (corrupted versions of correct geometries, such as those obtained from low-level calculations). Our starting point was the idea that the correct geometry is the best description of the target property. Hence, to incorporate information of the correct, GeoTMI aims to maximize mutual information between three variables: the correct and the corrupted geometries and the property. GeoTMI also explicitly updates the corrupted input to approach the correct geometry as it passes through the GNN layers, contributing to more effective denoising. We investigated the performance of the proposed method using 3D GNNs for three prediction tasks: molecular properties, a chemical reaction property, and relaxed energy in a heterogeneous catalytic system. Our results showed consistent improvements in accuracy across various tasks, demonstrating the effectiveness and robustness of GeoTMI.
\end{abstract}

\section{Introduction}

Neural networks have been actively applied to various fields of molecular and quantum chemistry \cite{Schtt2019, Manzhos2020, Dral2020, Park2022}.
Several input representations, such as the SMILES string and graph-based representations, are employed to predict quantum chemical properties \cite{David2020, Jiang2021}.
In particular, graph neural networks (GNNs), which operate on molecular graphs by updating the representation of each atom via message-passing based on chemical bonds, have achieved great success in many molecular property prediction tasks \cite{Schnet, Yang2019, Xiong2019, Zhou2020, Moon2022, Kim2022}.

However, as many quantum chemical properties depend on molecular geometries, typical GNNs without 3D geometric information have limitations in their accuracy.
In this respect, GNNs utilizing 3D information have recently achieved state-of-the-art accuracy \cite{Schnet,DimeNet,SphereNet,EGNN,Equiformer,transformer-M,SCN}.
Despite of their impressive accuracy, the usage of the 3D input geometry is often infeasible in real-world applications, limiting the 3D GNNs' applicability \cite{QM9M, 3DInfomax, GraphMVP, molecule3D, OC20}.
Therefore, it is natural to train machine learning models to make predictions with relatively easy-to-obtain geometries.
Several studies have investigated the use of easy-to-obtain geometry as input, and it has been empirically confirmed that such geometry can be leveraged to accurately predict target properties \cite{QM9M,molecule3D,OC20}.
Yet, theoretical basis for fully exploiting such easy-to-obtain geometries to predict accurate target properties remains to be established.

This study proposes a novel training framework, namely ``\textbf{Geo}metric denoising for \textbf{T}hree-term \textbf{M}utual \textbf{I}nformation maximization (GeoTMI)'', which employs a denoising process to accurately predict quantum chemical properties using easy-to-obtain geometries.
Throughout this paper, we denote the correct geometry as $X$, the easy-to-obtain geometry (regarded as the corrupted version of $X$) as $\tilde{X}$, and the target property as $Y$.
Various previous studies have been conducted on denoising approaches, such as a denoising autoencoder (DAE) \cite{ExtractDAE,StackDAE,JointDAE,SimpleGNN,PretrainingNoisyNode,DDM}.
When it comes to predicting quantum chemical properties, the predominant focus of denoising techniques has been on improving the prediction accuracy starting from $X$. GeoTMI, however, aims at improving the prediction accuracy starting from $\tilde{X}$.
GeoTMI also explicitly updates the input geometry of $\tilde{X}$ to approach $X$ as it passes through the GNN layers, thereby contributing to more effective denoising.
Furthermore, GeoTMI incorporates an auxiliary objective that predicts $Y$ from $X$, allowing it to capture the task-relevant information and ultimately maximize the mutual information (MI) between the three terms of $X$, $\tilde{X}$, and $Y$.
The theoretical derivations in this study provide further support for this approach.

GeoTMI offers the advantage of being model-agnostic and easy to integrate into existing GNN architectures.
Thus, in this study, we aimed to validate the effectiveness of GeoTMI on different GNN architectures and for various prediction tasks (the nine other molecular properties of the QM9 \cite{QM9}, a chemical reaction property of Grambow's dataset \cite{Grambow2020}, and relaxed energy in a heterogeneous catalytic system of the Open Catalyst 2020 (OC20) dataset \cite{OC20}).
We evaluated the performance of GeoTMI by comparing it to baselines trained only with $\tilde{X}$ and $Y$ using multiple 3D GNNs.
GeoTMI showed consistent accuracy improvements for all the target properties tested.
In particular, in our experiment on the IS2RE task of the OC20, GeoTMI achieved greater performance improvements than another denoising method, Noise Nodes \cite{SimpleGNN}, demonstrating the superiority of GeoTMI.
Overall, our findings demonstrate that GeoTMI can make accurate and robust predictions with easy-to-obtain geometries. Code is available on \href{https://anonymous.4open.science/r/GeoTMI-1DF5}{Github}.


\section{Related Works}
\subsection{Predicting high-level properties from easy-to-obtain geometry}

Recently, several deep learning approaches have aimed to predict high-level properties from an easy-to-obtain geometry for accurate yet fast predictions in real-world applications.
For instance, Molecule3D benchmark \cite{molecule3D} aims to improve the applicability of existing 3D models by developing machine learning models that predict 3D geometry.
These models predicted 3D geometry using 2D graph information that can be easily obtained and were evaluated using ETKDG \cite{ETKDG} from RDKit \cite{rdkit} as a baseline.

There have been attempts to predict high-level properties, starting with a geometry that can be quickly obtained by conventional methods, rather than machine learning methods.
\citet{QM9M} adopted Merck molecular force field (MMFF) \cite{MMFF94} geometries as the starting point, to predict density functional theory (DFT) \cite{DFT} properties of the molecules in the QM9 dataset.
In chemical reactions, \citet{DimeReaction} exploited reactant and product geometries to assess reaction barrier height rather than reactant and transition state (TS) geometries; because obtaining the TS geometry is computationally challenging.
In addition, \citet{OC20} proposed the Open Catalyst challenge.
In this challenge, the IS2RE task
uses initial structures (IS) for geometry optimization to predict the relaxed energies (RE) of the corresponding relaxed structures (RS).
In this case, the IS and the RE can be mapped into easy-to-obtain geometries and high-level properties, respectively.
Various approaches have been proposed to address this challenge  \cite{Equiformer,SimpleGNN,gemnet}.

GeoTMI shares the same goal as these previous works. However, it is important to emphasize that we propose a training framework based on a theoretical basis that possesses the capacity to be applicable to various tasks, rather than being limited to a specific task.

\subsection{Denoising approaches in GNN}
Denoising is a commonly used approach for representation learning by recovering correct data from corrupted data.
Previous studies have shown that models can learn desirable representations by mapping from a corrupted data manifold to a correct data manifold.
Traditional denoising auto-encoders (DAEs) employed a straightforward procedure of recovering a correct data from corrupted data thus maximizes the MI between correct data and its representation \cite{ExtractDAE, StackDAE}.
Recently, several studies in GNNs have adopted denoising strategies for representation learning and robustness during training \cite{PretrainingNoisyNode,DDM,GNNNoisePT,GNSNoise,GNNNoise}.
For instance, Noisy Nodes \cite{SimpleGNN}, which primary aim is addressing oversmoothing in GNNs, used denoising noisy node information as an auxiliary task, resulting in improved performance in property prediction.
Additionally, LaGraph \cite{LaGraph} leveraged predictive self-supervised learning of GNNs to predict the intractable latent graph that represents semantic information of an observed graph, by introducing a surrogate loss originated from image representation learning \cite{Noise2Same}.
While typical denoising approaches focus on learning representations of expensive $X$, GeoTMI aims to learn higher-quality representations for $\tilde{X}$, lying on a geometrically corrupted data manifold, to predict $Y$.
For this purpose, GeoTMI adopts the maximization of the three-term MI between $X$, $\tilde{X}$, and $Y$ with theoretical basis.

\subsection{Invariant 3D GNNs for quantum chemical properties}
In the field of chemistry, GNNs utilizing 3D geometric information have shown promising performance in predicting quantum chemical or systematic properties \cite{Schnet,DimeNet,SphereNet,EGNN,gemnet,Schnetpack,DimeNet++,ComENet}.
Since target physical quantities, such as an energy, are invariant to alignments of a molecule, 3D GNN models utilize roto-translational invariant 3D information as their inputs \cite{pmlr-v139-schutt21a, jing2020learning}.
As a representative example, a distance matrix guarantees the invariance because the roto-translational transformation does not vary distances.
SchNet \cite{Schnet, Schnetpack} and EGNN \cite{EGNN} are proper examples of utilizing the distance matrix.
The former exploits the radial basis function based on the distance matrix, while the latter uses distance information directly on the GNN message-passing scheme.
In DimeNet++ \cite{DimeNet++}, along with the distance matrix, bond angles are also available as invariant 3D information.
In addition, ComENet \cite{ComENet} and SphereNet \cite{SphereNet} introduced dihedral angles in addition to the distance and bond angle information.
Recently, several approaches such as Equiformer \cite{Equiformer} explicitly considered irreducible representations to construct roto-translational equivariant neural networks \cite{SE3Transformer,SEGNN}.
Our evaluation showed that GeoTMI is model-agnostic, hence can be easily applied to various 3D GNNs.

\section{Method}

\begin{figure}[ht]
\vskip 0.2in
\vspace{-0.2in}
\begin{center}
\centerline{\includegraphics[width=\columnwidth]{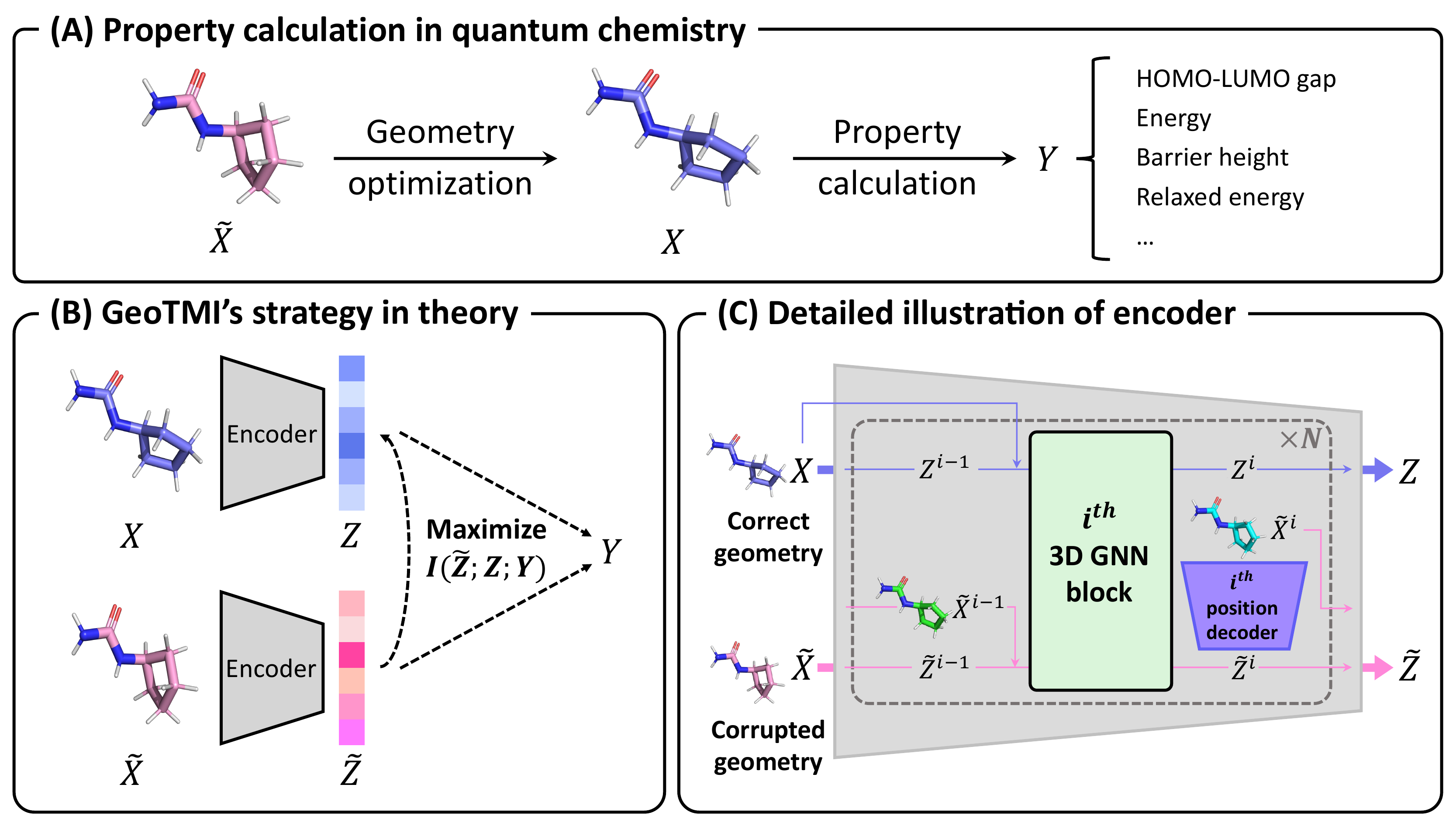}
}
\caption{
(A) Physical relationship between $\tilde{X}$, $X$, and $Y$ in property calculation in quantum chemistry.
(B) Schematic illustration of GeoTMI's strategy in theory, where objective is maximizing three-term MI, $I(\tilde{Z};Z;Y)$.
(C) Detailed illustration of encoder architecture in practical strategy of GeoTMI.
Training process employs both blue and pink lines, while inference process utilizes only pink line.
All molecular geometries were
plotted using PyMOL \cite{pymol}.}
\label{fig:1}
\end{center}
\vskip -0.2in
\end{figure}

In this section, we describe the overall framework of our proposed GeoTMI with theoretical background.
First, in \cref{sec:problem_setup}, we introduce the problem setting, i.e., the physical relationships required to predict a property from a corrupted geometry.
Then, in \cref{sec:training_objective}, we introduce our training objective for three-term MI, which differs from the objective of typical supervised learning.
Since MI itself is intractable, we derive a tractable loss for the training objective in \cref{sec:tractable_loss_derivation}.
Finally, in \cref{sec:overall_framework}, we illustrate the practical application of GeoTMI framework in the training and inference processes.

\subsection{Problem setup}
\label{sec:problem_setup}
We first introduce physical relationship between our data: corrupted geometry, $\tilde{X}$, correct geometry, $X$, and quantum chemical property, $Y$.
Our training data, $\mathcal{D}$, consist of observed samples $(\tilde{x},x,y)$ from the triplet of three random variables $(\tilde {X}, X, Y)\sim q(\tilde{X}, X, Y)$.
We assume that these three variables are interlinked through a Markov chain $\tilde{X}\to X\to Y$.
In our problem setting, $\tilde{Z}$ and $Z$ denote representations of $\tilde{X}$ and $X$, respectively, whose probability distributions are parameterized by $\theta$, $p_\theta(\tilde{Z},Z,Y)$.

Within the standard computational chemistry process, $Y$ is obtained from the correct geometry $X$, which is acquired through geometric optimization of $\tilde{X}$ as shown in \cref{fig:1}(A).
This process naturally gives rise to the Markov chain assumption, which suggests that $X$ encapsulates all the essential information for $Y$.
We can establish two assumptions employing the underlying physical relationship between the variables as an inductive bias.
First, there exists a higher quality of information pertaining to $Y$ within $X$ compared to $\tilde{X}$, or, more precisely, the MI between $X$ and $Y$ that is equal to or greater than the MI between $\tilde{X}$ and $Y$.
This naturally follows from the property of conditional independence between non-adjacent states in a Markov chain.
Second, the data distribution of $Y$ is solely dependent on $X$, irrespective of the presence of $\tilde{X}$, implying that $\tilde{X}$ and $Y$ are conditional independent given $X$.

The goal of GeoTMI is to obtain a proper representation $\tilde{Z}$ in predicting $Y$, by aligning it into $Z$ that contains more enriching information for $Y$.
GeoTMI differs to self-supervised learning by emphasizing a specialized representation that is tailored to the target property.
Also, acquisition of physical relationship between the variables as inductive bias leads to a higher quality representation than focusing on predicting $Y$ using $\tilde{X}$ alone.

\subsection{Training objective}
\label{sec:training_objective}
We propose a training framework for learning a proper representation for predicting $Y$ from $\tilde{X}$, which can be done by maximizing the MI between the variables, $I_{\theta}(\tilde{Z};Y)$.
This is somewhat similar to the objective of the general supervised learning which predicts $Y$ from $\tilde{X}$.
However, the training only with $\tilde{X}$ to predict $Y$ could be erroneous because there is no guarantee a model utilize the proper information resided in both $\tilde{X}$ and $X$.

If we introduce $Z$, one can express $I_{\theta}(\tilde{Z};Y)$ as following:
\begin{equation}
\label{eq:MI_tildeZ_Y}
I_{\theta}(\tilde{Z};Y)=I_{\theta}(\tilde{Z};Y|Z)+I_{\theta}(\tilde{Z};Z;Y).
\end{equation}

In \cref{eq:MI_tildeZ_Y}, the conditional MI $I_{\theta}(\tilde{Z};Y|Z)$ implies undesirable information of $\tilde{Z}$ in predicting $Y$ that is not relevant to $Z$.
This is a direct counterpart to the physical inductive bias in the previous section, where $X$ is sufficient information for the prediction of $Y$, and thus should be minimized to zero in the optimal case (see \cref{app:conditionally_independent_random_variables}).
Maximizing $I_\theta(\tilde{Z};Y)$ while maintaining the zero inductive bias is ideal, but non-trivial and challenging. 
To address this, we propose a straightforward solution by introducing the inductive bias as a regularization term, which reformulates our initial objective into maximization of the three-term MI,
\begin{equation}
I_\theta(\tilde{Z};Z;Y)=I_\theta(\tilde{Z};Y)-I_\theta(\tilde{Z};Y\vert Z). \nonumber
\end{equation}

\subsection{Tractable loss derivation}
\label{sec:tractable_loss_derivation}
In general, MI is not tractable, and accurately estimating it is another challenging task.
Thus, we have derived a tractable lower bound of the three-term MI, allowing us to practically maximize it (see \cref{fig:1}(B)).

\begin{proposition}
\label{prop:MI_LB}
$I(\tilde{Z};Z;Y)\ge \mathrm{LB}+H(Y)$ for any triplets of random variables $(\tilde{Z}, Z, Y)$, 
where $\mathrm{LB}=-H(Y|Z)-\frac{1}{2}H(Y|\tilde{Z})-\frac{1}{2}H(Z|\tilde{Z})$.
\end{proposition}

Since $H(Y)$ is constant term in respect of model parameters, we have the three distincts optimization targets: (1) property from corrupted representation $H(Y|\tilde{Z})$, (2) property from correct representation $H(Y|Z)$, and (3) reconstruction to correct representation $H(Z|\tilde{Z})$.

The conditional entropy term related to the property is estimated by a parameterized distribution $p_{\pi_{1}}$ based on the positiveness of KL divergence:
\begin{align}
\nonumber
-H(Y|Z)-H(Y|\tilde{Z})&=\mathbb{E}_{p_{\theta}(\tilde{Z},Z,Y)}\left[\log p_\theta(Y|\tilde{Z})+\log p_\theta(Y|Z)\right] \\\nonumber
&\ge\mathbb{E}_{p_{\theta}(\tilde{Z},Z,Y)}\left[\log p_{\pi_{1}}(Y|\tilde{Z})+\log p_{\pi_{1}}(Y|Z)\right] \\\nonumber
&\sim-\sum\left(\mathcal{L}\left(y,h_{\pi_1}(\tilde{z})\right) +\mathcal{L}\left(y,h_{\pi_1}(z)\right)\right). \nonumber
\end{align}

Here, we introduce property predictor $h_{\pi_1}$ which is parameterized by $\pi_1$.
Similarly, the other term is estimated by a parameterized distribution $p_{\pi_2}$,
\begin{align}
\nonumber
-H(Z|\tilde{Z})&=\mathbb{E}_{p_{\theta}(\tilde{Z},Z,Y)}\left[\log p_\theta(Z|\tilde{Z})\right] \\ \nonumber
&\ge \mathbb{E}_{p_\theta(\tilde{Z},Z,Y)}\left[\log p_{{\pi_{2}}}(Z|\tilde{Z})\right] \\\nonumber
&\sim -\sum \mathcal{L}(z,\hat{g}_{\pi_{2}}(\tilde{z})). \nonumber
\end{align}
Here, $\hat{g}_{\pi_2}: \tilde{Z}\to Z$ denotes a parametric decoder for information flows.
Since $Z$ is a parameterized variable which is not optimal in an initial training stage, the optimization could be unstable.
If we assume the encoder $f_\theta:X\to Z$ is continuous bijective, we could introduce a surrogate loss of decoding $\tilde{Z}$ into $X$,

\begin{equation}
\nonumber
\mathbb{E}_{p_{\theta}(\tilde{Z},Z,Y)}\left[\mathcal{L}\left(f_{\theta}^{-1}(Z),f_{\theta}^{-1}\circ \hat{g}_{\pi_2}(\tilde{Z})\right)\right]=\mathbb{E}_{p_{\theta}(\tilde{Z},Z,Y)}\left[\mathcal{L}\left(X,g_{\pi_2}(\tilde{Z})\right)\right],
\end{equation}

where $g_{\pi_2}=f_{\theta}^{-1}\circ\hat{g}_{\pi_2} : \tilde{Z}\to X$ denotes a decoder reconstructing $X$ rather than $Z$.
It still maximizes MI between $\tilde{Z}$ and $Z$, in that the continuous and bijective mapping does not change the MI.
In summary, the training process is about finding optimal model parameters $\theta$, $\pi_1$, and $\pi_2$ to minimize the following:
\begin{equation}
\nonumber
\mathbb{E}_{p_\theta(Z,\tilde{Z},Y)}
\bigg[
\underbrace{\mathcal{L}\big(Y,h_{\pi_1}(\tilde{Z})\big)}_{\mathcal{L}_\mathrm{y,corrupted}}+
\underbrace{\mathcal{L}\big(Y,h_{\pi_1}(Z)\big)}_{\mathcal{L}_\mathrm{y,correct}}+
\underbrace{\mathcal{L}\big(X,g_{\pi_2}(\tilde{Z})\big)}_{\mathcal{L}_d}
\bigg].
\end{equation}

The tractable loss function comprises the three terms: $\mathcal{L}_{\mathrm{y,corrupted}}$, $\mathcal{L}_{\mathrm{y,correct}}$, and $\mathcal{L}_{d}$.
We refer to $\mathcal{L}_\mathrm{y}$ as the property prediction loss and $\mathcal{L}_d$ as the denoising loss.
We chose the absolute error for the loss function $\mathcal{L}$.
The proof of \cref{prop:MI_LB} and details of the denoising loss are described in \cref{Appendix.A}.

\subsection{Overall framework}
\label{sec:overall_framework}
The proposing framework comprises the encoder, predictor, and decoder.
The encoder maps molecular geometries to their representations, while the predictor estimates target properties, and the decoder restores the molecular geometries.
The encoder design involves 3D GNN layers for both $X$ and $\tilde{X}$, sharing model parameters.
It is appropriate approach because $X$ and $\tilde{X}$ belong to the same data modality.
In addition, the encoder for $\tilde{X}$ includes explicit position update layers that are inspired by the geometry optimization process.
The effect of intermediate geometries as input is studied in \cref{app:easy2obtain2corrupted}.
The practical model architecture including the encoder design is depicted in \Cref{fig:1}(C).

In practice, an auxiliary loss is introduced as an add-on for the denoising loss to softly guide the position update toward $X$.
We will refer to this as gradual denoising loss, which measures the difference between each updated geometry and the corresponding linearly interpolated target geometry. 
The details and ablation study of this are in \cref{app:denoising_method_ablation}.
We chose the same architecture of the position update layer for the decoder.

During training, $\tilde{X}$ and $X$ are mapped to $\tilde{Z}$ and $Z$ respectively.
The property prediction loss is computed based on the results from $\tilde{Z}$ and $Z$, while the denoising loss involves reconstruction of $X$ from $\tilde{Z}$.
In inference process, only $\tilde{Z}$ encoded by $\tilde{X}$ is used for property prediction.
It is noteworthy that GeoTMI introduces a novel representation learning approach that leverages $X$ and $Z$ for robust property prediction, and its effectiveness lies in not requiring $X$ and $Z$ during the inference process.

\section{Experiments}
We have tried to demonstrate the effectiveness of GeoTMI in providing a new solution to the infeasibility of high-level 3D geometry, rather than focusing on the performance of the state-of-the-art GNN architecture itself.
Thus, in this section, we have focused on showing the applicability of GeoTMI to a variety of GNN architectures and its effectiveness in predicting properties in various areas of chemistry.
The tasks and architectures tested were selected based on computational cost and memory efficiency, as well as model performance.
All experiments were conducted using RTX 2080 Ti GPU with 12 GB of memory, RTX 3080 Ti GPU with 12 GB of memory, or RTX A4000 GPU with 16 GB of memory.
GNN models were trained on a single GPU, except for those in the IS2RE task of OC20, where we used eight RTX A4000 GPUs.

\subsection{Molecular property prediction}
\label{subsection:Molecular property prediction}

Predicting molecular properties is crucial to various fields in chemistry.
The QM9 \cite{QM9} is widely used benchmark dataset for molecular property prediction comprised of 134k molecule information; each molecule consists of at most nine heavy atoms (C, N, O, and F).
Each data sample contains optimized geometry and more than 10 corresponding DFT properties.

This study focuses on the QM9$_\mathrm{M}$ \cite{QM9M} task which predicts DFT properties using the MMFF geometry.
The QM9$_\mathrm{M}$ dataset originated from the QM9 differs only in the molecular geometry; each geometry herein has been obtained with additional MMFF optimization starting with the corresponding geometry in the QM9.
Here, the MMFF geometry is regarded as a relatively easy-to-obtain geometry compared to the DFT geometry.

\textbf{Training setup}.
This study employed the following three GNNs using distinct 3D information to demonstrate the effectiveness of GeoTMI: EGNN \cite{EGNN} (implementation follows \cite{delfta}), SchNet \cite{Schnetpack}, and DimeNet++ \cite{DimeNet++}.
We appended the position update to the SchNet and DimeNet++ to ensure that the denoising process can be applied to them in the same manner in the coordinate update of the EGNN.
To train the models, we considered the DFT geometry from the QM9 dataset as $X$, and the corresponding MMFF geometry from QM9$_\mathrm{M}$ dataset as $\tilde{X}$.

Molecular 2D graph information is similar to MMFF geometry information in that it is also more readily available than DFT geometry.
There have been many attempts to predict accurate molecular properties from 2D graphs alone \cite{transformer-M, OGB-LSC,  GPS++}.
Recently, \citet{transformer-M} developed the Transformer-M model, which can utilize both 2D graph and 3D geometry information in training to predict molecular properties with high accuracy using only 2D graphs.
To compare the usefulness of 2D graphs and MMFF geometries as easy-to-obtain inputs, we evaluated the prediction performance of the Transformer-M model on 2D graph inputs without pre-training.
Note that the Transformer-M model reported their performance using $X$ based on pre-training in the original paper.

For all tested models, we used 100,000, 18,000, and 13,000 molecular data for training, validation, and testing, respectively, as in previous work by \citet{EGNN}.
The detailed hyperparameters of each model are introduced in \Cref{QM9_hyperparameter}.

\textbf{Results on molecular property prediction}.
\Cref{tab:QM9_results} shows the prediction accuracy according to input types and models.
Results for SchNet and DimeNet++ are shown in \cref{A.2}.
GeoTMI achieved  performance improvements across all properties and models.
For example, GeoTMI resulted in accuracy improvements of 7.0$\sim$27.1\% for EGNN, as shown in \Cref{tab:QM9_results}.
Meanwhile, Transformer-M trained using both 2D graphs and $X$ resulted in accuracy improvements of -15$\sim$21\% compared to the same model trained using 2D graphs only.
Despite the similar prediction performance of Transformer-M and EGNN based on $X$, it is noteworthy that for most properties, the Transformer-M models using 2D graphs for prediction were less accurate than the 3D GNNs tested.
This result implies that while both the MMFF geometry and the molecular the 2D graph are easy-to-obtain inputs, the MMFF geometry contains more useful information for learning the relationship between molecules and their quantum chemical properties.
Furthermore, we conducted additional experiments for three properties ($\mu$, $R^2$, and $U_0$) using scaffold-based splitting, a methodology that offers a more realistic and demanding setting for evaluating out-of-distribution (OOD) generalization (see \Cref{A.4}).
Once again, GeoTMI consistently improved its prediction performance, highlighting the robustness of GeoTMI.

\begin{table*}[h]
\caption{
MAEs for QM9's properties.
The best performance among the models that do not use $X$ in the inference (Infer.) process is shown in bold.
The values of Transformer-M using $X$ were borrowed from \citet{transformer-M}.
The performance of GeoTMI integrated with SchNet and DimeNet++ is provided in \Cref{A.2}.
}
\label{tab:QM9_results}
\vskip 0.15in
\centering
\centering
\resizebox{\textwidth}{!}{
\begin{tabular}{lcccccccccc}
\toprule
Methods & 
\begin{tabular}[c]{@{}c@{}} Input type\\
(Train $/$ Infer.) \end{tabular} &
\begin{tabular}[c]{@{}c@{}} $U_0$ \\
($\mathrm{meV}$) \end{tabular} & 
\begin{tabular}[c]{@{}c@{}} $\mu$ \\
($\mathrm{D}$) \end{tabular} & 
\begin{tabular}[c]{@{}c@{}} $\alpha$ \\
($\mathrm{Bohr^3}$) \end{tabular} &  
\begin{tabular}[c]{@{}c@{}} $\epsilon_\mathrm{HOMO}$ \\ ($\mathrm{meV}$) \end{tabular} & 
\begin{tabular}[c]{@{}c@{}} $\epsilon_\mathrm{LUMO}$ \\ ($\mathrm{meV}$) \end{tabular} & 
\begin{tabular}[c]{@{}c@{}} GAP \\ ($\mathrm{meV}$) \end{tabular} & 
\begin{tabular}[c]{@{}c@{}} $R^2$ \\ ($\mathrm{Bohr^2}$) \end{tabular} & 
\begin{tabular}[c]{@{}c@{}} $C_v$ \\ ($\frac{cal}{mol\cdot K}$) \end{tabular} & 
\begin{tabular}[c]{@{}c@{}} ZPVE \\ ($\mathrm{meV}$) \end{tabular} \\
\midrule
Transformer-M \cite{transformer-M}  & $X / X$ & 14.8 & - & - & 26.5 & 23.8 & - & - & - & - \\
EGNN & $X / X$ & 12.9 & 0.0350 & 0.0759 & 31.2 & 26.6 & 51.1 & 0.130 & 0.0336 & 1.59 \\
\midrule
Transformer-M & 2D $/$ 2D & 38.2 & 0.309 & 0.171 & 53.6 & 52.5 & 77.1 & 11.4 & 0.0669 & 4.79 \\
Transformer-M & 2D, $X$ $/$ 2D & 43.9 & 0.245 & 0.160 & 48.7 & 46.3 & 68.4 & 10.3 & 0.0683 & 3.85 \\
\midrule
EGNN & $\tilde{X} / \tilde{X}$ & 17.4 &	0.133 & 0.125 & 38.4 & 34.4 & 58.0 & 5.60 &0.0445 & 1.97 \\
EGNN + GeoTMI & $X, \tilde{X} / \tilde{X}$ & \textbf{14.5} & \textbf{0.100} & \textbf{0.105} & \textbf{35.7} & \textbf{31.2} & \textbf{53.2} & \textbf{4.08} & \textbf{0.0407} & \textbf{1.76} \\
\midrule
\midrule
\multicolumn{2}{l}{Improvements by GeoTMI (\%)} & 16.7 & 24.8 & 16.0 & 7.03 & 9.30 & 8.28 & 27.1 & 8.54 & 10.7 \\
\bottomrule
\end{tabular}
}
\vskip -0.1in
\end{table*}
 
\subsection{Reaction property prediction}
\label{subsection:Reaction property prediction}
A chemical reaction is a process in which reactant molecules convert to product molecules, passing through their TSs.
Predicting properties related to the reaction is important for understanding the nature of the chemistry \cite{Grambow2020}.
The barrier height, as one of the reaction properties, is defined as the energy difference between the geometry of the TS, $X^{TS}$, and the geometry of the reactant, $X^{R}$.
Commonly, optimizing a $X^{TS}$ utilizes both $X^R$ and the geometry of the product $X^P$.
However, this optimization process is typically resource-intensive, requiring approximately 10 times more computational resources than optimizing $X^R$ or $X^P$ individually \cite{TS-expensive}.
Thus, predicting accurate barrier height without $X^{TS}$ is necessary to reduce computational costs.

From this point of view, we focused on the task to predict DFT calculation-based barrier heights using $X^R$ and $X^P$ for the elementary reaction of the gas phase, as reported by \citet{DimeReaction}.
In contrast to most molecular properties, the property is a function of not just a single molecular geometry, but ($X^R$, $X^{TS}$), which can be interpreted as an optimized version of ($X^R$, $X^P$).
Thus, we considered that $X := (X^R, X^{TS})$ and $\tilde{X} := (X^R, X^P)$ in this task.

We used two datasets, released by \citet{Grambow2020}, for comparison with the previous work.
The first dataset consists of unimolecular reactions, namely CCSD(T)-UNI.
The second dataset, B97-D3, has 16,365 reactions.

\textbf{Training setup}. 
\citet{Spiekermann2022} proposed two models for predicting the barrier height using the 2D and 3D information of $\tilde{X}$, respectively.
They used D-MPNN for 2D GNN and DimeNet++ for 3D GNN, which will be referred to as the 2D D-MPNN model and the 3D DimeReaction (DimeRxn), respectively.
Here, the DimeRxn trained with $\tilde{X}$ showed lower performance than the 2D D-MPNN because the 3D GNNs were sensitive to the noise in the input geometry, as pointed out in another study \cite{TSDiff}.
Thus, our method, which removes the noise, can be useful for DimeRxn.

For the DimeRxn, we also adopted EGNN’s coordinate update scheme as a decoder to predict correct geometries. 
We trained the D-MPNN model, and the DimeRxn models without and with GeoTMI for CCSD(T)-UNI and B97-D3 datasets. 
The used data split and augmentation were the same as in the previous work by \citet{Spiekermann2022}.
In particular, we note that scaffold splitting was used on the datasets to evaluate the OOD generalization ability of the model.
The hyperparameters used are described in \Cref{rxn_hyperparameter}.

\begin{table*}[h]
\caption{
MAEs for predicted reaction barrier ($\mathrm{kcal/mol}$).
The best performance among the models that do not use $X$ in the inference (Infer.) process is shown in bold.
}
\label{tab:rxn_results}
\vskip 0.15in
\begin{center}
\begin{tabular}{lccc}
\toprule
\multirow{2}{*}{Methods} & Input type  & \multicolumn{2}{c}{Dataset} \\
\cmidrule{3-4}
 & (Train $/$ Infer.) & CCSD(T)-UNI & B97-D3 \\
\toprule
\textrm{DimeRxn}  & $X / X$ & 2.38 &  1.92  \\
\midrule
D-MPNN & 2D / 2D & 4.59 & 4.91   \\
\midrule
\textrm{DimeRxn}  & $\tilde{X} / \tilde{X}$  & 6.03 &  7.32  \\
\textrm{DimeRxn + GeoTMI} & $X, \tilde{X} / \tilde{X}$  & \textbf{3.90} &  \textbf{4.17}  \\
\midrule
\midrule
\multicolumn{2}{l}{Improvements by GeoTMI (\%)} & 35.3 & 43.0  \\
\bottomrule
\end{tabular}
\end{center}
\vskip -0.1in
\end{table*}

\textbf{Results on reaction property prediction}.
\Cref{tab:rxn_results} shows the results of prediction accuracy according to input types and models.
The DimeRxn trained with $X$ has the best prediction performance for all methods, while DimeRxn trained with $\tilde{X}$ has the worst prediction performance.
The result supports that DimeRxn is highly dependent on the quality of input geometry, as previously mentioned.
Thus, as we expected, GeoTMI, which is developed for learning a proper representation for predicting $Y$ from  $\tilde{X}$, induced accuracy improvements of 35.4\% and 43.0\% than DimeRxn without GeoTMI in terms of MAE for the CCSD(T)-UNI and B97-D3 datasets, respectively.
The results show that it outperforms the 2D D-MPNN model, again demonstrating the usefulness of the 3D easy-to-obtain geometry with GeoTMI, which is identified in the previous section.

\subsection{IS2RE prediction}
The OC20 dataset contains data consisting of the slab called a catalyst and molecules called adsorbates for each of the systems.
There are more atoms and a wider variety of atom types compared to previously studied datasets.
In detail, the dataset contains more than 460k pairs of IS, RS, and RE.
We focus on the IS2RE task, which is to predict the RE using the IS.
From the perspective of computational chemistry, the RS are obtained through costly quantum chemical calculations based on the IS.
Thus, in this task, we considered IS as $\tilde{X}$ and RS as $X$.

\textbf{Training setup}.
We adopted the Equiformer model \cite{Equiformer} to evaluate the effectiveness of GeoTMI in the IS2RE task.
The Equiformer model achieved state-of-the-art performance by using Noisy Nodes, where the IS2RS auxiliary task was integrated with the IS2RE task.
We note that the hyperparameters used are the same as in the previous work, except for the number of transformer blocks to train each model, due to the limitation of our computational resources.
We refer to the model trained only on the IS2RE task without Noisy Nodes as the baseline model, namely Equiformer*.
We performed a comparative analysis of three training frameworks: (1) Equiformer*, (2) Equiformer* + Noisy Nodes, and (3) Equiformer* + GeoTMI.
Thus, this evaluation with Equiformer can show the effectiveness of GeoTMI on the baseline model while allowing for comparison with Noisy Nodes.

To implement the Equiformer with GeoTMI, we followed much of the original Equiformer paper.
First, we used the same Noisy Nodes data augmentation.
Second, we used a similar node-level auxiliary loss for the IS2RS task.
The auxiliary loss predicts the node-level difference between target positions and noisy inputs, which corresponds to the denoising loss of $\mathcal{L}_d$.
The different points of the ``Equiformer* + GeoTMI'' compared to the ``Equiformer* + Noisy Nodes'' are as follows.
The noisy positions were explicitly updated by passing through GNN layers.
The detailed objective here is to calculate the difference between the updated noisy positions and the linearly interpolated target positions at each GNN layer, which we refer to as the gradual denoising loss in our paper.
In addition, we incorporated an auxiliary task that predicts the RE from the RS, denoted as $\mathcal{L}_{y,\text{correct}}$, which ultimately facilitates the training process of maximizing the three-term MI.


\begin{table*}[h]
\caption{
Results on the OC20 IS2RE test set with different methods based on Equiformer architectures \cite{Equiformer}.
The Equiformer* denotes a model that reduces the number of transformer blocks from 18 to 4 while keeping other hyperparameters the same.
The best performance among the Equiformer* models is shown in bold, and its improvement rate is shown in the last row.
}
\label{tab:OC20_results}
\vskip 0.15in
\begin{center}
\resizebox{\textwidth}{!}
{
\begin{tabular}{lcccccccccc}
\toprule
 & \multicolumn{5}{c}{Energy MAE (eV) $\downarrow$} & \multicolumn{5}{c}{EwT (\%) $\uparrow$} \\
\cmidrule(lr){2-6} \cmidrule(lr){7-11}
\multirow{2}{*}{Methods} & ID & \begin{tabular}[c]{@{}l@{}}OOD\\ Ads\end{tabular} &
\begin{tabular}[c]{@{}l@{}}OOD\\ Cat\end{tabular} &
\begin{tabular}[c]{@{}l@{}}OOD\\ Both\end{tabular} & Average &
ID & \begin{tabular}[c]{@{}l@{}}OOD\\ Ads\end{tabular} &
\begin{tabular}[c]{@{}l@{}}OOD\\ Cat\end{tabular} &
\begin{tabular}[c]{@{}l@{}}OOD\\ Both\end{tabular} & Average 
\\
\midrule
Equiformer + Noisy Nodes \cite{Equiformer} & 0.417 & 0.548 & 0.425 & 0.474 & 0.466 & 7.71 & 3.70 & 7.15 & 4.07 & 5.66  \\
\midrule
Equiformer* & 0.515 & 0.651 & 0.531 & 0.603 & 0.575 & 4.81 & 2.50 & 4.45 & 2.86 & 3.66   \\
Equiformer* + Noisy Nodes & 0.449 & 0.606 & 0.460 & 0.540 & 0.513 & 6.47 & 3.04 & 5.83 & 3.52 & 4.72  \\
Equiformer* + GeoTMI & \textbf{0.425} & \textbf{0.583} & \textbf{0.440} & \textbf{0.521} & \textbf{0.492}  & \textbf{7.60} & \textbf{3.86} & \textbf{6.97} & \textbf{4.03} & \textbf{5.62}  \\
\midrule
\midrule
Improvement (\%) & 17.6 & 10.5 & 17.1 & 13.7 & 14.4 & 58.0 & 54.4 & 56.6 & 40.9 & 53.8   \\
\bottomrule
\end{tabular}
}
\end{center}
\vskip -0.1in
\end{table*}

\textbf{Results on IS2RE with Noisy Nodes}.
We have summarized the IS2RE results in \Cref{tab:OC20_results}.
To evaluate each method, the MAE of the RE prediction using IS and the energy within a threshold (EwT), the percentage in which the MAE of the predicted energy is within 0.02 eV, are used.
Both Noisy Nodes and GeoTMI show performance improvements over the baseline Equiformer*, but GeoTMI achieves better performance gains across all metrics.
Despite the improvements, the prediction performance with GeoTMI is still lower than the original model with 18 transformer blocks in most cases in terms of MAE.
However, the prediction performance is similar to the original model for EwT and even better for OOD Ads.

\subsection{Ablation study}
\begin{table*}[b]
\caption{
Ablation study for GeoTMI.
BH and PU denote reaction barrier height and position update, respectively.
Prediction accuracy is compared in terms of MAE.
The most degraded results are underlined.
}
\label{tab:ablation_rxn_results}
\vskip 0.15in
\begin{center}
\begin{tabular}{lllcccc}
\toprule
Dataset & Property & Unit & GeoTMI & w/o $\mathcal{L}_d$ & w/o PU & w/o $\mathcal{L}_{\mathrm{y},\mathrm{correct}}$ \\
\midrule
\multirow{4}{*}{QM9 + QM9$_\mathrm{M}$} & $U_0$ & $\mathrm{meV}$ & 14.5 & \underline{21.0} & 14.2  & 15.2 \\
 & $R^2$ & $\mathrm{Bohr}^2$ & 4.08 & \underline{6.43} & 4.12  & 4.43 \\
 & $C_v$ & $\mathrm{cal/mol\cdot K}$ & 0.0407 & \underline{0.0503} & 0.0410  & 0.0401 \\
 & $\mu$ & $\mathrm{D}$ & 0.0997 & \underline{0.127} & 0.111  & 0.110 \\
\midrule
CCSD(T)-UNI & BH & $\mathrm{kcal/mol}$ & 3.90 &  4.41 & \underline{5.77}  & 4.27 \\
B97-D3 & BH & $\mathrm{kcal/mol}$ & 4.17 & 4.49 & \underline{7.19} &  4.45 \\
\bottomrule
\end{tabular}
\end{center}
\vskip -0.1in
\end{table*}

GeoTMI uses a combination of $\mathcal{L}_d$, $\mathcal{L}_{\mathrm{y},\mathrm{correct}}$, and the position update to improve the accuracy of predicting quantum chemical properties using $\tilde{X}$.
To verify an individual contribution of each component of GeoTMI, we conducted ablation studies. 
\Cref{tab:ablation_rxn_results} shows that all strategies are individually meaningful to reduce prediction error regardless of the properties.
In this experiment, it is noteworthy that training without either $\mathcal{L}_{\mathrm{y},\mathrm{correct}}$ or $\mathcal{L}_d$ is no longer maximizing the lower bound of the three-term MI.
The prediction performance of these models performs worse than trained models using GeoTMI except for $C_v$. 
The results imply that our proposed three-term MI maximization is key in prediction performance based on $\tilde{X}$.
Additionally, the table shows that position update, introduced by our intuition, is a key component for the BH prediction and helps increase overall prediction performance.

\section{Conclusion and Limitations}
In this study, we propose GeoTMI, a novel training framework designed to exploit easy-to-obtain geometry for accurate prediction of quantum chemical properties. 
The proposed framework is based on the Markov chain assumption and the theoretical basis that maximizes the mutual information (MI) between property, correct and corrupted geometries, mitigating the degradation in accuracy resulting from the use of the corrupted geometry.
To achieve this, GeoTMI incorporates a denoising process to effectively address the inherent challenges associated with acquiring correct 3D geometry.
In particular, we introduced the position update in the denoising process and gradual denoising loss to enhance the efficacy of the training process.

We have verified that GeoTMI consistently improves the prediction performance of 3D GNNs for three benchmark datasets.
Nevertheless, there are several limitations in this work. 
First, GeoTMI addresses the inductive bias by incorporating a soft regularization approach instead of directly vanishing it to zero. 
Second, we could not perform an extensive optimal hyperparameter search due to a lack of computational resources.
However, our consistent experimental results on various tasks showed the effectiveness and robustness of the GeoTMI.
In this light, we envision that the GeoTMI becomes a new solution to solve the practical infeasibility of high-cost 3D geometry in many other chemistry fields.

\section{Acknowledgement}
This work was supported by the Korea Environmental Industry and Technology Institute (Grant No. RS202300219144), the Technology Innovation Program funded by the Ministry of Trade, Industry \& Energy, MOTIE, Korea (Grant No. 20016007), and the Ministry of Science and ICT, Korea (Grant No. RS-2023-00257479).

\bibliography{neurips_2023}

\begin{thebibliography}{57}
\providecommand{\natexlab}[1]{#1}
\providecommand{\url}[1]{\texttt{#1}}
\expandafter\ifx\csname urlstyle\endcsname\relax
  \providecommand{\doi}[1]{doi: #1}\else
  \providecommand{\doi}{doi: \begingroup \urlstyle{rm}\Url}\fi

\bibitem[Sch\"{u}tt et~al.(2019)Sch\"{u}tt, Gastegger, Tkatchenko, M\"{u}ller,
  and Maurer]{Schtt2019}
K.~T. Sch\"{u}tt, M.~Gastegger, A.~Tkatchenko, K.-R. M\"{u}ller, and R.~J.
  Maurer.
\newblock Unifying machine learning and quantum chemistry with a deep neural
  network for molecular wavefunctions.
\newblock \emph{Nature Communications}, 10\penalty0 (1), November 2019.
\newblock \doi{10.1038/s41467-019-12875-2}.
\newblock URL \url{https://doi.org/10.1038/s41467-019-12875-2}.

\bibitem[Manzhos and Carrington(2020)]{Manzhos2020}
Sergei Manzhos and Tucker Carrington.
\newblock Neural network potential energy surfaces for small molecules and
  reactions.
\newblock \emph{Chemical Reviews}, 121\penalty0 (16):\penalty0 10187--10217,
  October 2020.
\newblock \doi{10.1021/acs.chemrev.0c00665}.
\newblock URL \url{https://doi.org/10.1021/acs.chemrev.0c00665}.

\bibitem[Dral(2020)]{Dral2020}
Pavlo~O. Dral.
\newblock Quantum chemistry in the age of machine learning.
\newblock \emph{The Journal of Physical Chemistry Letters}, 11\penalty0
  (6):\penalty0 2336--2347, March 2020.
\newblock \doi{10.1021/acs.jpclett.9b03664}.
\newblock URL \url{https://doi.org/10.1021/acs.jpclett.9b03664}.

\bibitem[Park et~al.(2022)Park, Han, Kim, and Choi]{Park2022}
Sanggil Park, Herim Han, Hyungjun Kim, and Sunghwan Choi.
\newblock Machine learning applications for chemical reactions.
\newblock \emph{Chemistry {\textendash} An Asian Journal}, 17\penalty0 (14),
  May 2022.
\newblock \doi{10.1002/asia.202200203}.
\newblock URL \url{https://doi.org/10.1002/asia.202200203}.

\bibitem[David et~al.(2020)David, Thakkar, Mercado, and Engkvist]{David2020}
Laurianne David, Amol Thakkar, Roc{\'{\i}}o Mercado, and Ola Engkvist.
\newblock Molecular representations in {AI}-driven drug discovery: a review and
  practical guide.
\newblock \emph{Journal of Cheminformatics}, 12\penalty0 (1), September 2020.
\newblock \doi{10.1186/s13321-020-00460-5}.
\newblock URL \url{https://doi.org/10.1186/s13321-020-00460-5}.

\bibitem[Jiang et~al.(2021)Jiang, Wu, Hsieh, Chen, Liao, Wang, Shen, Cao, Wu,
  and Hou]{Jiang2021}
Dejun Jiang, Zhenxing Wu, Chang-Yu Hsieh, Guangyong Chen, Ben Liao, Zhe Wang,
  Chao Shen, Dongsheng Cao, Jian Wu, and Tingjun Hou.
\newblock Could graph neural networks learn better molecular representation for
  drug discovery? {A} comparison study of descriptor-based and graph-based
  models.
\newblock \emph{Journal of Cheminformatics}, 13\penalty0 (1), February 2021.
\newblock \doi{10.1186/s13321-020-00479-8}.
\newblock URL \url{https://doi.org/10.1186/s13321-020-00479-8}.

\bibitem[Sch\"{u}tt et~al.(2018{\natexlab{a}})Sch\"{u}tt, Sauceda, Kindermans,
  Tkatchenko, and M\"{u}ller]{Schnet}
K.~T. Sch\"{u}tt, H.~E. Sauceda, P.-J. Kindermans, A.~Tkatchenko, and K.-R.
  M\"{u}ller.
\newblock {SchNet} {\textendash} a deep learning architecture for molecules and
  materials.
\newblock \emph{The Journal of Chemical Physics}, 148\penalty0 (24):\penalty0
  241722, June 2018{\natexlab{a}}.
\newblock \doi{10.1063/1.5019779}.
\newblock URL \url{https://doi.org/10.1063/1.5019779}.

\bibitem[Yang et~al.(2019)Yang, Swanson, Jin, Coley, Eiden, Gao, Guzman-Perez,
  Hopper, Kelley, Mathea, Palmer, Settels, Jaakkola, Jensen, and
  Barzilay]{Yang2019}
Kevin Yang, Kyle Swanson, Wengong Jin, Connor Coley, Philipp Eiden, Hua Gao,
  Angel Guzman-Perez, Timothy Hopper, Brian Kelley, Miriam Mathea, Andrew
  Palmer, Volker Settels, Tommi Jaakkola, Klavs Jensen, and Regina Barzilay.
\newblock Analyzing learned molecular representations for property prediction.
\newblock \emph{Journal of Chemical Information and Modeling}, 59\penalty0
  (8):\penalty0 3370--3388, July 2019.
\newblock \doi{10.1021/acs.jcim.9b00237}.
\newblock URL \url{https://doi.org/10.1021/acs.jcim.9b00237}.

\bibitem[Xiong et~al.(2019)Xiong, Wang, Liu, Zhong, Wan, Li, Li, Luo, Chen,
  Jiang, and Zheng]{Xiong2019}
Zhaoping Xiong, Dingyan Wang, Xiaohong Liu, Feisheng Zhong, Xiaozhe Wan, Xutong
  Li, Zhaojun Li, Xiaomin Luo, Kaixian Chen, Hualiang Jiang, and Mingyue Zheng.
\newblock Pushing the boundaries of molecular representation for drug discovery
  with the graph attention mechanism.
\newblock \emph{Journal of Medicinal Chemistry}, 63\penalty0 (16):\penalty0
  8749--8760, August 2019.
\newblock \doi{10.1021/acs.jmedchem.9b00959}.
\newblock URL \url{https://doi.org/10.1021/acs.jmedchem.9b00959}.

\bibitem[Zhou et~al.(2020)Zhou, Cui, Hu, Zhang, Yang, Liu, Wang, Li, and
  Sun]{Zhou2020}
Jie Zhou, Ganqu Cui, Shengding Hu, Zhengyan Zhang, Cheng Yang, Zhiyuan Liu,
  Lifeng Wang, Changcheng Li, and Maosong Sun.
\newblock Graph neural networks: A review of methods and applications.
\newblock \emph{{AI} Open}, 1:\penalty0 57--81, 2020.
\newblock \doi{10.1016/j.aiopen.2021.01.001}.
\newblock URL \url{https://doi.org/10.1016/j.aiopen.2021.01.001}.

\bibitem[Moon et~al.(2022)Moon, Zhung, Yang, Lim, and Kim]{Moon2022}
Seokhyun Moon, Wonho Zhung, Soojung Yang, Jaechang Lim, and Woo~Youn Kim.
\newblock {PIGNet}: a physics-informed deep learning model toward generalized
  drug{\textendash}target interaction predictions.
\newblock \emph{Chemical Science}, 13\penalty0 (13):\penalty0 3661--3673, 2022.
\newblock \doi{10.1039/d1sc06946b}.
\newblock URL \url{https://doi.org/10.1039/d1sc06946b}.

\bibitem[Kim et~al.(2022)Kim, Kim, and Kim]{Kim2022}
Jun~Hyeong Kim, Hyeonsu Kim, and Woo~Youn Kim.
\newblock Effect of molecular representation on deep learning performance for
  prediction of molecular electronic properties.
\newblock \emph{Bulletin of the Korean Chemical Society}, 43\penalty0
  (5):\penalty0 645--649, March 2022.
\newblock \doi{10.1002/bkcs.12516}.
\newblock URL \url{https://doi.org/10.1002/bkcs.12516}.

\bibitem[Gasteiger et~al.(2020{\natexlab{a}})Gasteiger, Groß, and
  G\"{u}nnemann]{DimeNet}
Johannes Gasteiger, Janek Groß, and Stephan G\"{u}nnemann.
\newblock Directional message passing for molecular graphs, 2020{\natexlab{a}}.
\newblock URL \url{https://arxiv.org/abs/2003.03123}.

\bibitem[Liu et~al.(2021{\natexlab{a}})Liu, Wang, Liu, Zhang, Oztekin, and
  Ji]{SphereNet}
Yi~Liu, Limei Wang, Meng Liu, Xuan Zhang, Bora Oztekin, and Shuiwang Ji.
\newblock Spherical message passing for {3D} graph networks,
  2021{\natexlab{a}}.
\newblock URL \url{https://arxiv.org/abs/2102.05013}.

\bibitem[Satorras et~al.(2021)Satorras, Hoogeboom, and Welling]{EGNN}
V\'{\i}ctor~Garcia Satorras, Emiel Hoogeboom, and Max Welling.
\newblock E(n) equivariant graph neural networks.
\newblock In Marina Meila and Tong Zhang, editors, \emph{Proceedings of the
  38th International Conference on Machine Learning}, volume 139 of
  \emph{Proceedings of Machine Learning Research}, pages 9323--9332. PMLR,
  18--24 Jul 2021.
\newblock URL \url{https://proceedings.mlr.press/v139/satorras21a.html}.

\bibitem[Liao and Smidt(2022)]{Equiformer}
Yi-Lun Liao and Tess Smidt.
\newblock Equiformer: Equivariant graph attention transformer for {3D}
  atomistic graphs, 2022.
\newblock URL \url{https://arxiv.org/abs/2206.11990}.

\bibitem[Luo et~al.(2022)Luo, Chen, Xu, Zheng, Liu, Wang, and
  He]{transformer-M}
Shengjie Luo, Tianlang Chen, Yixian Xu, Shuxin Zheng, Tie-Yan Liu, Liwei Wang,
  and Di~He.
\newblock One transformer can understand both {2D} \& {3D} molecular data,
  2022.

\bibitem[Zitnick et~al.(2022)Zitnick, Das, Kolluru, Lan, Shuaibi, Sriram,
  Ulissi, and Wood]{SCN}
C.~Lawrence Zitnick, Abhishek Das, Adeesh Kolluru, Janice Lan, Muhammed
  Shuaibi, Anuroop Sriram, Zachary Ulissi, and Brandon Wood.
\newblock Spherical channels for modeling atomic interactions, 2022.
\newblock URL \url{https://arxiv.org/abs/2206.14331}.

\bibitem[Lu et~al.(2019)Lu, Wang, and Zhang]{QM9M}
Jianing Lu, Cheng Wang, and Yingkai Zhang.
\newblock Predicting molecular energy using force-field optimized geometries
  and atomic vector representations learned from an improved deep tensor neural
  network.
\newblock \emph{Journal of Chemical Theory and Computation}, 15\penalty0
  (7):\penalty0 4113--4121, May 2019.
\newblock \doi{10.1021/acs.jctc.9b00001}.
\newblock URL \url{https://doi.org/10.1021/acs.jctc.9b00001}.

\bibitem[St\"{a}rk et~al.(2021)St\"{a}rk, Beaini, Corso, Tossou, Dallago,
  G\"{u}nnemann, and Liò]{3DInfomax}
Hannes St\"{a}rk, Dominique Beaini, Gabriele Corso, Prudencio Tossou, Christian
  Dallago, Stephan G\"{u}nnemann, and Pietro Liò.
\newblock {3D Infomax} improves gnns for molecular property prediction.
\newblock 2021.
\newblock \doi{10.48550/ARXIV.2110.04126}.
\newblock URL \url{https://arxiv.org/abs/2110.04126}.

\bibitem[Liu et~al.(2021{\natexlab{b}})Liu, Wang, Liu, Lasenby, Guo, and
  Tang]{GraphMVP}
Shengchao Liu, Hanchen Wang, Weiyang Liu, Joan Lasenby, Hongyu Guo, and Jian
  Tang.
\newblock Pre-training molecular graph representation with {3D} geometry,
  2021{\natexlab{b}}.
\newblock URL \url{https://arxiv.org/abs/2110.07728}.

\bibitem[Xu et~al.(2021)Xu, Luo, Zhang, Xu, Xie, Liu, Dickerson, Deng, Nakata,
  and Ji]{molecule3D}
Zhao Xu, Youzhi Luo, Xuan Zhang, Xinyi Xu, Yaochen Xie, Meng Liu, Kaleb
  Dickerson, Cheng Deng, Maho Nakata, and Shuiwang Ji.
\newblock Molecule3{D}: A benchmark for predicting {3D} geometries from
  molecular graphs, 2021.
\newblock URL \url{https://arxiv.org/abs/2110.01717}.

\bibitem[Chanussot et~al.(2021)Chanussot, Das, Goyal, Lavril, Shuaibi, Riviere,
  Tran, Heras-Domingo, Ho, Hu, Palizhati, Sriram, Wood, Yoon, Parikh, Zitnick,
  and Ulissi]{OC20}
Lowik Chanussot, Abhishek Das, Siddharth Goyal, Thibaut Lavril, Muhammed
  Shuaibi, Morgane Riviere, Kevin Tran, Javier Heras-Domingo, Caleb Ho, Weihua
  Hu, Aini Palizhati, Anuroop Sriram, Brandon Wood, Junwoong Yoon, Devi Parikh,
  C.~Lawrence Zitnick, and Zachary Ulissi.
\newblock {Open Catalyst} 2020 ({OC}20) dataset and community challenges.
\newblock \emph{{ACS} Catalysis}, 11\penalty0 (10):\penalty0 6059--6072, May
  2021.
\newblock \doi{10.1021/acscatal.0c04525}.
\newblock URL \url{https://doi.org/10.1021/acscatal.0c04525}.

\bibitem[Vincent et~al.(2008)Vincent, Larochelle, Bengio, and
  Manzagol]{ExtractDAE}
Pascal Vincent, Hugo Larochelle, Yoshua Bengio, and Pierre-Antoine Manzagol.
\newblock Extracting and composing robust features with denoising autoencoders.
\newblock In \emph{Proceedings of the 25th international conference on Machine
  learning}, pages 1096--1103, 2008.

\bibitem[Vincent et~al.(2010)Vincent, Larochelle, Lajoie, Bengio, and
  Manzagol]{StackDAE}
Pascal Vincent, Hugo Larochelle, Isabelle Lajoie, Yoshua Bengio, and
  Pierre-Antoine Manzagol.
\newblock Stacked denoising autoencoders: Learning useful representations in a
  deep network with a local denoising criterion.
\newblock \emph{Journal of Machine Learning Research}, 11\penalty0
  (110):\penalty0 3371--3408, 2010.
\newblock URL \url{http://jmlr.org/papers/v11/vincent10a.html}.

\bibitem[Zhou et~al.(2014)Zhou, Arpit, Nwogu, and Govindaraju]{JointDAE}
Yingbo Zhou, Devansh Arpit, Ifeoma Nwogu, and Venu Govindaraju.
\newblock Is joint training better for deep auto-encoders?, 2014.
\newblock URL \url{https://arxiv.org/abs/1405.1380}.

\bibitem[Godwin et~al.(2021)Godwin, Schaarschmidt, Gaunt, Sanchez-Gonzalez,
  Rubanova, Veličković, Kirkpatrick, and Battaglia]{SimpleGNN}
Jonathan Godwin, Michael Schaarschmidt, Alexander Gaunt, Alvaro
  Sanchez-Gonzalez, Yulia Rubanova, Petar Veličković, James Kirkpatrick, and
  Peter Battaglia.
\newblock Simple {GNN} regularisation for {3D} molecular property prediction \&
  beyond, 2021.
\newblock URL \url{https://arxiv.org/abs/2106.07971}.

\bibitem[Zaidi et~al.(2022)Zaidi, Schaarschmidt, Martens, Kim, Teh,
  Sanchez-Gonzalez, Battaglia, Pascanu, and Godwin]{PretrainingNoisyNode}
Sheheryar Zaidi, Michael Schaarschmidt, James Martens, Hyunjik Kim, Yee~Whye
  Teh, Alvaro Sanchez-Gonzalez, Peter Battaglia, Razvan Pascanu, and Jonathan
  Godwin.
\newblock Pre-training via denoising for molecular property prediction.
\newblock \emph{arXiv preprint arXiv:2206.00133}, 2022.

\bibitem[Liu et~al.(2022)Liu, Guo, and Tang]{DDM}
Shengchao Liu, Hongyu Guo, and Jian Tang.
\newblock Molecular geometry pretraining with {SE}(3)-invariant denoising
  distance matching, 2022.
\newblock URL \url{https://arxiv.org/abs/2206.13602}.

\bibitem[Ramakrishnan et~al.(2014)Ramakrishnan, Dral, Rupp, and von
  Lilienfeld]{QM9}
Raghunathan Ramakrishnan, Pavlo~O. Dral, Matthias Rupp, and O.~Anatole von
  Lilienfeld.
\newblock Quantum chemistry structures and properties of 134 kilo molecules.
\newblock \emph{Scientific Data}, 1\penalty0 (1), August 2014.
\newblock \doi{10.1038/sdata.2014.22}.
\newblock URL \url{https://doi.org/10.1038/sdata.2014.22}.

\bibitem[Grambow et~al.(2020)Grambow, Pattanaik, and Green]{Grambow2020}
Colin~A. Grambow, Lagnajit Pattanaik, and William~H. Green.
\newblock Deep learning of activation energies.
\newblock \emph{The Journal of Physical Chemistry Letters}, 11\penalty0
  (8):\penalty0 2992--2997, March 2020.
\newblock \doi{10.1021/acs.jpclett.0c00500}.
\newblock URL \url{https://doi.org/10.1021/acs.jpclett.0c00500}.

\bibitem[Riniker and Landrum(2015)]{ETKDG}
Sereina Riniker and Gregory~A. Landrum.
\newblock Better informed distance geometry: Using what we know to improve
  conformation generation.
\newblock \emph{Journal of Chemical Information and Modeling}, 55\penalty0
  (12):\penalty0 2562--2574, November 2015.
\newblock \doi{10.1021/acs.jcim.5b00654}.
\newblock URL \url{https://doi.org/10.1021/acs.jcim.5b00654}.

\bibitem[Contributors(2021)]{rdkit}
RDKit Contributors.
\newblock {RDK}it: Open-source cheminformatics, 2021.
\newblock URL \url{http://www.rdkit.org/}.
\newblock Accessed: 2023-03-19.

\bibitem[Halgren(1996)]{MMFF94}
Thomas~A. Halgren.
\newblock Merck molecular force field. {I}. {B}asis, form, scope,
  parameterization, and performance of {MMFF}94.
\newblock \emph{Journal of Computational Chemistry}, 17\penalty0
  (5-6):\penalty0 490--519, April 1996.
\newblock \doi{10.1002/(sici)1096-987x(199604)17:5/6<490::aid-jcc1>3.0.co;2-p}.
\newblock URL
  \url{https://doi.org/10.1002/(sici)1096-987x(199604)17:5/6<490::aid-jcc1>3.0.co;2-p}.

\bibitem[Kohn and Sham(1965)]{DFT}
W.~Kohn and L.~J. Sham.
\newblock Self-consistent equations including exchange and correlation effects.
\newblock \emph{Phys. Rev.}, 140:\penalty0 A1133--A1138, Nov 1965.
\newblock \doi{10.1103/PhysRev.140.A1133}.
\newblock URL \url{https://link.aps.org/doi/10.1103/PhysRev.140.A1133}.

\bibitem[Spiekermann et~al.(2022{\natexlab{a}})Spiekermann, Pattanaik, and
  Green]{DimeReaction}
Kevin~A. Spiekermann, Lagnajit Pattanaik, and William~H. Green.
\newblock Fast predictions of reaction barrier heights: Toward coupled-cluster
  accuracy.
\newblock \emph{The Journal of Physical Chemistry A}, 126\penalty0
  (25):\penalty0 3976--3986, June 2022{\natexlab{a}}.
\newblock \doi{10.1021/acs.jpca.2c02614}.
\newblock URL \url{https://doi.org/10.1021/acs.jpca.2c02614}.

\bibitem[Gasteiger et~al.(2021)Gasteiger, Becker, and G{\"u}nnemann]{gemnet}
Johannes Gasteiger, Florian Becker, and Stephan G{\"u}nnemann.
\newblock Gemnet: Universal directional graph neural networks for molecules.
\newblock \emph{Advances in Neural Information Processing Systems},
  34:\penalty0 6790--6802, 2021.

\bibitem[Hu et~al.(2019)Hu, Liu, Gomes, Zitnik, Liang, Pande, and
  Leskovec]{GNNNoisePT}
Weihua Hu, Bowen Liu, Joseph Gomes, Marinka Zitnik, Percy Liang, Vijay Pande,
  and Jure Leskovec.
\newblock Strategies for pre-training graph neural networks.
\newblock \emph{arXiv preprint arXiv:1905.12265}, 2019.

\bibitem[Sanchez-Gonzalez et~al.(2020)Sanchez-Gonzalez, Godwin, Pfaff, Ying,
  Leskovec, and Battaglia]{GNSNoise}
Alvaro Sanchez-Gonzalez, Jonathan Godwin, Tobias Pfaff, Rex Ying, Jure
  Leskovec, and Peter Battaglia.
\newblock Learning to simulate complex physics with graph networks.
\newblock In \emph{International conference on machine learning}, pages
  8459--8468. PMLR, 2020.

\bibitem[Sato et~al.(2021)Sato, Yamada, and Kashima]{GNNNoise}
Ryoma Sato, Makoto Yamada, and Hisashi Kashima.
\newblock Random features strengthen graph neural networks.
\newblock In \emph{Proceedings of the 2021 SIAM International Conference on
  Data Mining (SDM)}, pages 333--341. SIAM, 2021.

\bibitem[Xie et~al.(2022)Xie, Xu, and Ji]{LaGraph}
Yaochen Xie, Zhao Xu, and Shuiwang Ji.
\newblock Self-supervised representation learning via latent graph prediction.
\newblock In \emph{International Conference on Machine Learning}, pages
  24460--24477. PMLR, 2022.

\bibitem[Xie et~al.(2020)Xie, Wang, and Ji]{Noise2Same}
Yaochen Xie, Zhengyang Wang, and Shuiwang Ji.
\newblock Noise2{S}ame: Optimizing a self-supervised bound for image denoising.
\newblock \emph{Advances in Neural Information Processing Systems},
  33:\penalty0 20320--20330, 2020.

\bibitem[Sch\"{u}tt et~al.(2018{\natexlab{b}})Sch\"{u}tt, Kessel, Gastegger,
  Nicoli, Tkatchenko, and M\"{u}ller]{Schnetpack}
K.~T. Sch\"{u}tt, P.~Kessel, M.~Gastegger, K.~A. Nicoli, A.~Tkatchenko, and
  K.-R. M\"{u}ller.
\newblock {SchNetPack}: A deep learning toolbox for atomistic systems.
\newblock \emph{Journal of Chemical Theory and Computation}, 15\penalty0
  (1):\penalty0 448--455, November 2018{\natexlab{b}}.
\newblock \doi{10.1021/acs.jctc.8b00908}.
\newblock URL \url{https://doi.org/10.1021/acs.jctc.8b00908}.

\bibitem[Gasteiger et~al.(2020{\natexlab{b}})Gasteiger, Giri, Margraf, and
  G\"{u}nnemann]{DimeNet++}
Johannes Gasteiger, Shankari Giri, Johannes~T. Margraf, and Stephan
  G\"{u}nnemann.
\newblock Fast and uncertainty-aware directional message passing for
  non-equilibrium molecules, 2020{\natexlab{b}}.
\newblock URL \url{https://arxiv.org/abs/2011.14115}.

\bibitem[Wang et~al.(2022)Wang, Liu, Lin, Liu, and Ji]{ComENet}
Limei Wang, Yi~Liu, Yuchao Lin, Haoran Liu, and Shuiwang Ji.
\newblock Com{EN}et: Towards complete and efficient message passing for {3D}
  molecular graphs.
\newblock \emph{arXiv preprint arXiv:2206.08515}, 2022.

\bibitem[Sch{\"u}tt et~al.(2021)Sch{\"u}tt, Unke, and
  Gastegger]{pmlr-v139-schutt21a}
Kristof Sch{\"u}tt, Oliver Unke, and Michael Gastegger.
\newblock Equivariant message passing for the prediction of tensorial
  properties and molecular spectra.
\newblock In Marina Meila and Tong Zhang, editors, \emph{Proceedings of the
  38th International Conference on Machine Learning}, volume 139 of
  \emph{Proceedings of Machine Learning Research}, pages 9377--9388. PMLR,
  18--24 Jul 2021.
\newblock URL \url{https://proceedings.mlr.press/v139/schutt21a.html}.

\bibitem[Jing et~al.(2020)Jing, Eismann, Suriana, Townshend, and
  Dror]{jing2020learning}
Bowen Jing, Stephan Eismann, Patricia Suriana, Raphael~JL Townshend, and Ron
  Dror.
\newblock Learning from protein structure with geometric vector perceptrons.
\newblock \emph{arXiv preprint arXiv:2009.01411}, 2020.

\bibitem[Fuchs et~al.(2020)Fuchs, Worrall, Fischer, and
  Welling]{SE3Transformer}
Fabian Fuchs, Daniel Worrall, Volker Fischer, and Max Welling.
\newblock {SE}(3)-{T}ransformers: 3{D} roto-translation equivariant attention
  networks.
\newblock \emph{Advances in Neural Information Processing Systems},
  33:\penalty0 1970--1981, 2020.

\bibitem[Brandstetter et~al.(2021)Brandstetter, Hesselink, van~der Pol,
  Bekkers, and Welling]{SEGNN}
Johannes Brandstetter, Rob Hesselink, Elise van~der Pol, Erik~J Bekkers, and
  Max Welling.
\newblock Geometric and physical quantities improve {E}(3) equivariant message
  passing, 2021.
\newblock URL \url{https://arxiv.org/abs/2110.02905}.

\bibitem[{Schr\"odinger, LLC}(2017)]{pymol}
{Schr\"odinger, LLC}.
\newblock The {PyMOL} molecular graphics system, version~2.0.
\newblock 2017.

\bibitem[Atz et~al.(2022)Atz, Isert, B\"{o}cker, Jim{\'{e}}nez-Luna, and
  Schneider]{delfta}
Kenneth Atz, Clemens Isert, Markus N.~A. B\"{o}cker, Jos{\'{e}}
  Jim{\'{e}}nez-Luna, and Gisbert Schneider.
\newblock {$\Delta$}-quantum machine-learning for medicinal chemistry.
\newblock \emph{Physical Chemistry Chemical Physics}, 24\penalty0
  (18):\penalty0 10775--10783, 2022.
\newblock \doi{10.1039/d2cp00834c}.
\newblock URL \url{https://doi.org/10.1039/d2cp00834c}.

\bibitem[Hu et~al.(2021)Hu, Fey, Ren, Nakata, Dong, and Leskovec]{OGB-LSC}
Weihua Hu, Matthias Fey, Hongyu Ren, Maho Nakata, Yuxiao Dong, and Jure
  Leskovec.
\newblock {OGB-LSC}: A large-scale challenge for machine learning on graphs,
  2021.
\newblock URL \url{https://arxiv.org/abs/2103.09430}.

\bibitem[Masters et~al.(2022)Masters, Dean, Klaser, Li, Maddrell-Mander,
  Sanders, Helal, Beker, Rampášek, and Beaini]{GPS++}
Dominic Masters, Josef Dean, Kerstin Klaser, Zhiyi Li, Sam Maddrell-Mander,
  Adam Sanders, Hatem Helal, Deniz Beker, Ladislav Rampášek, and Dominique
  Beaini.
\newblock {GPS++}: An optimised hybrid mpnn/transformer for molecular property
  prediction, 2022.

\bibitem[Palazzesi et~al.(2020)Palazzesi, Hermann, Grundl, Pautsch, Seeliger,
  Tautermann, and Weber]{TS-expensive}
Ferruccio Palazzesi, Markus~R. Hermann, Marc~A. Grundl, Alexander Pautsch,
  Daniel Seeliger, Christofer~S. Tautermann, and Alexander Weber.
\newblock {BIreactive}: A machine-learning model to estimate covalent warhead
  reactivity.
\newblock \emph{Journal of Chemical Information and Modeling}, 60\penalty0
  (6):\penalty0 2915--2923, April 2020.
\newblock \doi{10.1021/acs.jcim.9b01058}.
\newblock URL \url{https://doi.org/10.1021/acs.jcim.9b01058}.

\bibitem[Spiekermann et~al.(2022{\natexlab{b}})Spiekermann, Pattanaik, and
  Green]{Spiekermann2022}
Kevin Spiekermann, Lagnajit Pattanaik, and William~H. Green.
\newblock High accuracy barrier heights, enthalpies, and rate coefficients for
  chemical reactions.
\newblock \emph{Scientific Data}, 9\penalty0 (1), July 2022{\natexlab{b}}.
\newblock \doi{10.1038/s41597-022-01529-6}.
\newblock URL \url{https://doi.org/10.1038/s41597-022-01529-6}.

\bibitem[Kim et~al.(2023)Kim, Woo, and Kim]{TSDiff}
Seonghwan Kim, Jeheon Woo, and Woo~Youn Kim.
\newblock Diffusion-based generative {AI} for exploring transition states from
  {2D} molecular graphs, 2023.
\newblock URL \url{https://arxiv.org/abs/2304.12233}.

\bibitem[Bemis and Murcko(1996)]{Bemis}
Guy~W. Bemis and Mark~A. Murcko.
\newblock The properties of known drugs. 1. molecular frameworks.
\newblock \emph{Journal of Medicinal Chemistry}, 39\penalty0 (15):\penalty0
  2887--2893, January 1996.
\newblock \doi{10.1021/jm9602928}.
\newblock URL \url{https://doi.org/10.1021/jm9602928}.

\end{thebibliography}

\newpage
\appendix
\onecolumn
\section{Theoretical Basis}
\label{Appendix.A}
\subsection{Conditionally independence} 
\label{app:conditionally_independent_random_variables}

\begin{proposition}
\label{prop:C_independence2C_MI}
For any random variables $(\tilde{X}, X, Y)$, if $\tilde{X}$ and $Y$ are conditional independent given $X$, then $I(\tilde{X};Y\vert X)=0$.
\end{proposition}
\begin{proof}[Proof of \cref{prop:C_independence2C_MI}]
From the definition of conditional mutual information (MI), we start from the below.
\begin{align}
I(\tilde{X};Y|X)&=H(\tilde{X}|X)+H(Y|X)-H(\tilde{X},Y|X) \nonumber\\
&=\mathbb{E}_{p(\tilde{X},X,Y)}\left[-\log \frac{p(\tilde{X}|X)p(Y|X)}{p(\tilde{X},Y|X)}\right] \nonumber\\
&=\mathbb{E}_{p(\tilde{X},X,Y)}\left[-\log \frac{p(\tilde{X},X)p(X,Y)/p(X)^2}{p(\tilde{X},X,Y)/p(X)}\right] \nonumber\\
&=\mathbb{E}_{p(\tilde{X},X,Y)}\left[-\log \frac{p(\tilde{X},X)p(X,Y)}{p(\tilde{X},X,Y)p(X)}\right] \nonumber\\
&=\mathbb{E}_{p(\tilde{X},X,Y)}\left[-\log \frac{p(X,Y)/p(X)}{p(\tilde{X},X,Y)/p(\tilde{X},X)}\right] \nonumber\\
&=\mathbb{E}_{p(\tilde{X},X,Y)}\left[-\log \frac{p(Y|X)}{p(Y|\tilde{X},X)}\right] \nonumber\\
&=\mathbb{E}_{p(\tilde{X},X,Y)}\left[-\log \cancel{\frac{p(Y|X)}{p(Y|\tilde{X},X)}}\right] &(\text{conditional independence})\nonumber\\
&=0 \nonumber 
\end{align}
\end{proof}

\subsection{Lower bound of three-term MI}
\label{app:lower_bound_of_three_term_MI}
Here, we derive the lower bound of three-term MI described at \cref{prop:MI_LB}.

\begin{proof}[Proof of \cref{prop:MI_LB}]
We need to prove the following inequality first.
\begin{equation}
H(Z)-I(\tilde{Z};Z)\ge I(Z;Y)-I(\tilde{Z};Z;Y) \quad\forall \text{ random variables } \tilde{Z},Z,Y \nonumber
\end{equation}
\begin{align}
\text{LHS}&=H(Z|\tilde{Z})\nonumber\\
\text{RHS}&=I(Z;Y)-[I(Z;Y)-I(Z;Y|\tilde{Z})]\nonumber\\
&=I(Z;Y|\tilde{Z})\nonumber\\
&=\mathbb{E}_{p(\tilde{Z},Z,Y)}\left[-\log \frac{p(Z|\tilde{Z})p(Y|\tilde{Z})}{p(Z,Y|\tilde{Z})}\right]\nonumber\\
&=\mathbb{E}_{p(\tilde{Z},Z,Y)}\left[-\log \frac{p(Z|\tilde{Z})}{p(Z|\tilde{Z},Y)}\right]\nonumber\\
&=H(Z|\tilde{Z})-H(Z|\tilde{Z},Y)\nonumber
\end{align}
Since conditional entropy is non-negative, LHS$\ge$RHS.

By applying the above, we derive the following two inequalities:
\begin{align}
I(\tilde{Z};Z;Y)&\ge I(\tilde{Z};Z)+I(Z;Y)-H(Z)\nonumber\\
&=-[H(Z)-I(\tilde{Z};Z)]-[H(Y)-I(Z;Y)]+H(Y)\nonumber\\
&=-H(Z|\tilde{Z}) -H(Y|Z)+H(Y), \nonumber 
\end{align}
\begin{align}
I(\tilde{Z};Z;Y)&\ge I(\tilde{Z};Y)+I(Z;Y)-H(Y)\nonumber\\
&=-[H(Y)-I(\tilde{Z};Y)]-[H(Y)-I(Z;Y)]+H(Y)\nonumber\\
&=-H(Y|\tilde{Z}) -H(Y|Z)+H(Y). \nonumber
\end{align}
By adding the two inequalities, we derive a lower bound:
\begin{equation}
I(\tilde{Z};Z;Y)\ge H(Y)\underbrace{-H(Y\vert Z) -\frac{1}{2}H( Y\vert \tilde{Z}) -\frac{1}{2}H(Z\vert \tilde{Z})}_\mathrm{LB}. \nonumber
\end{equation}
\end{proof}
Though the coefficient of the lower bound is different for each term, our practical loss is calculated as follows:
\begin{equation}
\mathcal{L}_\mathrm{total}=\mathcal{L}_{\mathrm{y, corrupted}}+\mathcal{L}_{\mathrm{y, correct}}+ \lambda\mathcal{L}_{d}, \nonumber
\end{equation}
where $\mathcal{L}_{\mathrm{y, corrupted}}$, $\mathcal{L}_{\mathrm{y, correct}}$, and $\mathcal{L}_{d}$ correspond to $H( Y\vert \tilde{Z})$, $H(Y\vert Z)$, and $H(Z\vert \tilde{Z})$, respectively, and the coefficient $\lambda$ is adopted for a practical reason.
The searching space of $\lambda$ is described in \cref{Appendix.C}.

\subsection{Choice of denoising loss}
\label{app:choice_of_denoising_loss}
\textbf{Surrogate loss}.
We decode $X$ from $\tilde{Z}$ as a surrogate task for $H(Z\vert \tilde{Z})$.
Since an MI is invariant to continuous and bijective mappings, the surrogate loss to reconstruct $X$ can be obtained by transforming from continuous representation space to data space.
Although general 3D GNNs do not satisfy this requirement, we have empirically confirmed robust results in various chemistry tasks using several GNNs.

Nevertheless, it is necessary to explain the difference between maximizing $I(\tilde{Z};Z)$ and $I(X|\tilde{Z})$.
As shown in \Cref{fig:1}(A), various properties can be obtained from $X$, implying that $X$ contains information that is necessary to predict all the properties.
However, $Z$ partially contains the information of $X$ in that it is a representation of $X$.
For a specific property prediction task, the ideal situation would be for $Z$ to contain only the information necessary to accurately predict $Y$.
In this context, mapping $\tilde{Z}$ to $X$ instead of $Z$ may cause superfluous information which is irrelevant to $Y$.
But even so, it does no harm for our overall purpose of addressing the inductive bias of physical relationship.

\textbf{Geometric denoising loss}.
We incorporated a geometric denoising loss as a surrogate loss in order to maximize the MI. 
Specifically, we aimed to maintain the equivariance of $X$ under rotation or translation of $\tilde{X}$.
To achieve this, we employed the SE(3)-equivariant decoders in this study.

For the loss metric, we chose the MAE of the atom-pair distances. 
This choice was natural considering the importance of bond distances in molecular geometry compared to absolute atomic positions.
Specifically, the denoising metric $\mathcal{L}$ is calculated based on the $(i,j)$ atom-pair distances $d_{ij}$ and $\tilde{d}_{ij}$ of $x$ and $\tilde{x}$, respectively, which is 
\begin{equation}
\label{eq:denoising_metric}
\mathcal{L}\left(x, \tilde{x}\right)= \frac{1}{\vert \mathcal{E}\vert}\sum_{(i,j)\in\mathcal{E}} \vert d_{ij} - \tilde{d}_{ij}\vert. \\
\end{equation}
Here, $\mathcal{E}$ denotes a set of edges in a graph of $\tilde{x}$.
Thus, the practical denoising loss is calculated as follows:
\begin{equation}
\label{eq:denoising_loss}
\mathcal{L}_d=\frac{1}{\vert \mathcal{D}\vert}\sum_{(\tilde{x},x,y)\in\mathcal{D}}\mathcal{L}\big(x, g\left(f\left( \tilde{x}\right)\right)\big),
\end{equation}
where $f$ and $g$ denote an encoder and a decoder, respectively, and $\mathcal{D}$ is the dataset.
However, in case of the OC20 dataset, we computed the loss in atomic positions to align with the baseline model for comparisons.

To effectively minimize the denoising loss, we introduced an explicit position update at each layer, inspired by the geometry optimization of quantum chemical calculation. 
These position update layers facilitated the differentiation between the encoders of $\tilde{X}$ and $X$, despite sharing model parameters. 
It involved incorporating relatively small parameter additions to induce distinct mappings.

To achieve effective denoising and align the encoding process with the geometry optimization, we introduced a gradual denoising loss.
This additional loss term guides the position updates at each layer to exhibit directional behavior by forcing the updated geometry to lie within the linear interpolation between $\tilde{X}$ and $X$.
We compared the prediction performance according to the degree of corruption of the input geometry, and confirmed that the prediction performance improved as the geometry closer to $X$ was used (see \cref{app:easy2obtain2corrupted}).
It indirectly explains the reason for the introducing explicit position update and gradual loss.
More details on the loss form and ablation studies of gradual denoising are provided in \cref{app:denoising_method_ablation}.

\section{Further analyses} \label{Appendix.B}
\subsection{Utilization of interpolated geometries} \label{app:easy2obtain2corrupted}
In this study, we have assumed a Markov chain given that the correct geometry $X$ is optimized from the corrupted geometry $\tilde{X}$, and the quantum chemical property $Y$ is computed from $X$.
Since $X$ is more adjacent to $Y$ in the Markov chain, predictions of $Y$ based on $X$ should be more accurate than those based on $\tilde{X}$.
Furthermore, from a physical standpoint, the transition from $\tilde{X}$ to $X$ is a form of geometry optimization that could be perceived as a Markov chain. 
Building on these assumptions, the intermediate geometry during the optimization process naturally lies between $X$ and $\tilde{X}$ within the Markov chain.

In this section, we explored the possibility of using interpolated geometry $\frac{X+\tilde{X}}{2}$ as a surrogate intermediate geometry.
We expected that the following two statements will be satisfied.

\begin{itemize}
\item Regardless of the types of properties, the baseline (trained with $\tilde{X}$) will always have a higher MAE than the ground truth (trained with $X$).
\item The interpolated geometry between $\tilde{X}$ and $X$ is a less corrupted geometry than $\tilde{X}$.
Thus, MAEs of the model using the interpolated geometry are placed between those of the baseline and the ground truth.
\end{itemize}

For the demonstration, we used three individual prediction models, where the training input for each model was $X$, $\frac{X+\tilde{X}}{2}$, and $\tilde{X}$, respectively.
As expected, the results presented in \Cref{tab:input_noise,tab:input_noise_reaction} consistently show that the predictive accuracy of the model decreases as the level of corruption in the geometry increases.
This empirical evidence underscores the superiority of interpolated geometry over $\tilde{X}$ as an input for predicting various quantum chemical properties.

Thus, we can consider $\tilde{X}\to\frac{X+\tilde{X}}{2}\to X$ as a proxy of the geometry optimization process which is a Markov chain.
Inspired by these findings, we integrated the explicit position update scheme of EGNN \cite{EGNN} and introduced an additional loss that exploits the interpolated geometry in the position update step following geometry optimization.

\begin{table}[h]
\scriptsize
\caption{
MAEs for QM9's properties according to different inputs. Values were obtained using EGNN.
($\tilde{X}$+$X$)/2 denotes the interpolated geometry generated from the mean of the atom-pair distance of $\tilde{X}$ and $X$.
}
\label{tab:input_noise}
\vskip 0.15in
\begin{center}
\begin{small}
\begin{tabular}{lccccr}
\toprule
Target & Unit & $\tilde{X}$ & ($\tilde{X}$+$X$)/2 & $X$ \\
\midrule
  $U_0$ & $\mathrm{meV}$ & 17.4 &  14.2 & 12.9  \\
  $\mu$ & $\mathrm{D}$ & 0.133 & 0.0807 & 0.0350 \\
  $\alpha$     & $\mathrm{Bohr}^3$ & 0.125 & 0.0947 & 0.0759 \\
 $\epsilon_\mathrm{HOMO}$      & $\mathrm{meV}$ & 38.4 & 33.4 & 31.2  \\
 $\epsilon_\mathrm{LUMO}$      & $\mathrm{meV}$ & 34.4 & 28.9 & 26.6  \\
  GAP     & $\mathrm{meV}$ & 58.0 & 53.0 & 51.1 \\
  $R^2$     & $\mathrm{Bohr}^2$ & 5.60 & 2.92 & 0.130 \\
  $C_v$     & $\mathrm{cal/mol\cdot K}$ & 0.0445 & 0.0377 &  0.0336 \\
  ZPVE     &   $\mathrm{meV}$ & 1.97 & 1.74 & 1.59 \\
\bottomrule
\end{tabular}
\end{small}
\end{center}
\vskip -0.1in
\end{table}

\begin{table}[h]
\scriptsize
\caption{
MAEs (kcal/mol) of predicted barrier heights according to different inputs. Values were obtained using DimeReaction.
($\tilde{X}$+$X$)/2 denotes the pairwise interpolated geometry generated from the mean of the atom-pair distance of ($X^{R}$, $X^{P}$) and ($X^{R}$, $X^{TS}$).
}
\label{tab:input_noise_reaction}
\vskip 0.15in
\begin{center}
\begin{small}
\begin{tabular}{lcccr}
\toprule
Dataset  & $\tilde{X}$ & ($\tilde{X}$+$X$)/2 & $X$ \\
\midrule
  CCSD(T)-UNI  & 6.49 &  5.30 & 2.38  \\
  B97-D3 & 8.24 & 3.91 & 1.92 \\
\bottomrule
\end{tabular}
\end{small}
\end{center}
\vskip -0.1in
\end{table}

\subsection{Ablation study for denoising process}
In GeoTMI, the denoising process proceeds throughout all GNN layers.
The relationship between the number of GNN layers and the average absolute difference based on the atom-pair distance (D-MAE) between the correct geometry and the predicted geometry of the atomic positions at the last layer was investigated using the EGNN model (see \Cref{num_layer}).
We note that each EGNN was trained using only the denoising loss to confirm the denoising power alone.

\label{app:denoising_method_ablation}
\begin{figure}[ht]
\centering
\centerline{\includegraphics[width=10.0cm]{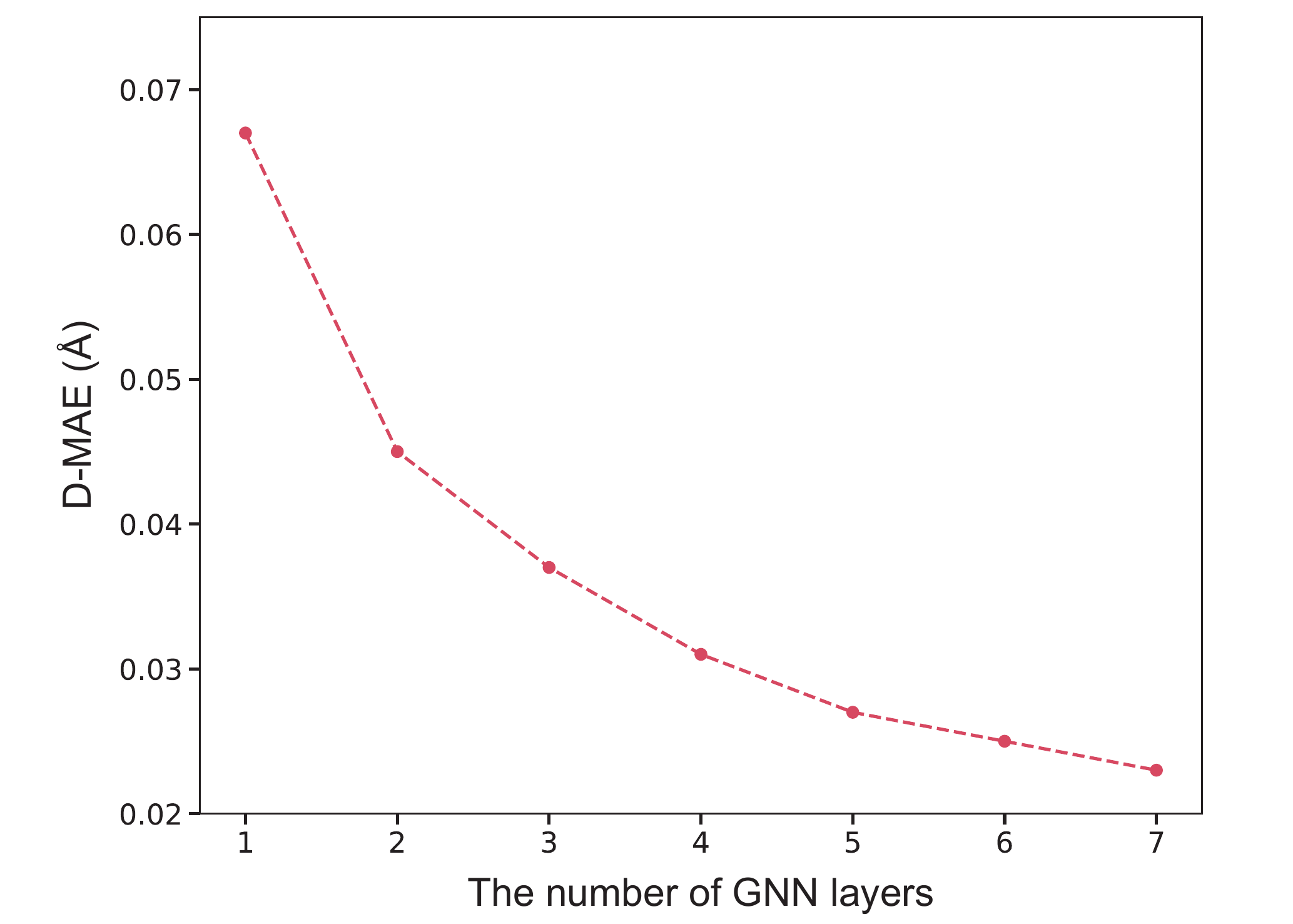}}
\caption{D-MAE according to the number of GNN layers of the EGNN. The D-MAE is measured using $X$ and the denoised $\tilde{X}$ derived from the last layer.}
\label{num_layer}
\end{figure}

In particular, the denoising process with the small number of GNN layers could not restore $X$.
We note that the D-MAE between the MMFF and the DFT geometries is 0.0712\AA.
If the denoising process couldn't restore $X$, the corresponding denoising loss can have a negative effect on learning by maintaining a large amount compared to other losses in the training process.
In this respect, we designed a gradual denoising loss, which is a slightly modified version of \cref{eq:denoising_loss} with a linearly interpolated target instead of $x$.
For each position update layer, a target varies linearly from $\tilde{x}$ to $x$, and the loss is calculated according to \cref{eq:denoising_metric}.
Specifically, the target distance of the $l$-th layer among a total of $L$ layers is as follows:
\begin{equation}
\bar{d}_{ij}^l=\frac{1}{L}\left(ld_{ij}+(L-l)\tilde{d}_{ij}\right). \nonumber
\end{equation}
\Cref{tab:gradual} shows that the gradual denoising process is useful to increase model performance.
Thus, we adopted the gradual denoising loss to predict quantum chemical properties.

\begin{table*}[htbp]
\scriptsize
\caption{Impact of each denoising task in terms of MAE.
We tested gradual denoising (``GeoTMI''), ``w/o Gradual denoising'', and ``Denoising last only'' for four properties.
When the gradual denoising was not used, the objective of all denoising processes of GNN layers is restoring $X$.
The ``Denoising last only'' means that denoising objective contains only the denoising loss of the last GNN layer without the losses of other GNN layers.
}
\label{tab:gradual}
\vskip 0.15in
\begin{center}
\begin{small}
\begin{tabular}{ccccc}
\toprule
Property & Unit & GeoTMI & w/o Gradual denoising  & Denoising last only\\
\midrule
$U_0$ & $\mathrm{meV}$ & 14.5  & 15.4 & 14.9\\
$R^2$ & $\mathrm{Bohr^2}$ & 4.08  & 4.22 & 4.97\\
$C_v$ & $\mathrm{cal/mol\cdot K}$ &0.0407  & 0.0413 & 0.0423\\
$\mu$ & $\mathrm{D}$ & 0.0997  & 0.100 &  0.132 \\
\bottomrule
\end{tabular}
\end{small}
\end{center}
\vskip -0.1in
\end{table*}

\subsection{Application of GeoTMI on SchNet and DimeNet++ for QM9 and QM9$_\mathrm{M}$ tasks} \label{A.2}

GeoTMI is applicable regardless of the 3D geometry model architectures.
In this regard, we tested the effectiveness of GeoTMI using two additional 3D GNNs: SchNet \cite{Schnetpack} and DimeNet++ \cite{DimeNet++}.
For comparison, we also identified the performance of QM9's properties prediction from $X$ for the two models (see \Cref{tab:QM9_various_models}).
The training, validation, and testing data were used as in \Cref{subsection:Molecular property prediction}.
  
\begin{table*}[h]
\caption{
MAEs for QM9's properties.
GeoTMI was tested with two different 3D GNNs: SchNet and DimeNet++.
}
\label{tab:QM9_various_models}
\vskip 0.15in
\centering
\resizebox{\textwidth}{!}{
\begin{tabular}{lcccccccccc}
\toprule
Approach & 
\begin{tabular}[c]{@{}c@{}} Input type\\
(Train $/$ Infer.) \end{tabular} &
\begin{tabular}[c]{@{}c@{}} $U_0$ \\
($\mathrm{meV}$) \end{tabular} & 
\begin{tabular}[c]{@{}c@{}} $\mu$ \\
($\mathrm{D}$) \end{tabular} & 
\begin{tabular}[c]{@{}c@{}} $\alpha$ \\
($\mathrm{Bohr^3}$) \end{tabular} &  
\begin{tabular}[c]{@{}c@{}} $\epsilon_\mathrm{HOMO}$ \\ ($\mathrm{meV}$) \end{tabular} & 
\begin{tabular}[c]{@{}c@{}} $\epsilon_\mathrm{LUMO}$ \\ ($\mathrm{meV}$) \end{tabular} & 
\begin{tabular}[c]{@{}c@{}} GAP \\ ($\mathrm{meV}$) \end{tabular} & 
\begin{tabular}[c]{@{}c@{}} $R^2$ \\ ($\mathrm{Bohr^2}$) \end{tabular} & 
\begin{tabular}[c]{@{}c@{}} $C_v$ \\ ($\frac{cal}{mol\cdot K}$) \end{tabular} & 
\begin{tabular}[c]{@{}c@{}} ZPVE \\ ($\mathrm{meV}$) \end{tabular} \\
\midrule
SchNet  & $X / X$ & 17.0 & 0.0391 & 0.0859 & 44.3 & 34.9 & 68.6 & 0.170 & 0.0313 & 1.67 \\
DimeNet++ & $X / X$ & 8.99 & 0.0382 & 0.0583 & 25.5 & 20.8 & 43.0 & 0.342 & 0.0255 & 1.30 \\
\midrule
SchNet & $\tilde{X} / \tilde{X}$ & 27.3 & 0.208 & 0.160 & 61.4 & 53.9 & 85.8 & 8.32 & 0.0595 & 2.57 \\
SchNet + GeoTMI & $X, \tilde{X} / \tilde{X}$ & 27.3 & 0.139 & 0.131 & 52.0 & 45.0 & 74.4 & 6.44 & 0.0566 & 2.38 \\
\midrule
DimeNet++ & $\tilde{X} / \tilde{X}$ & 16.9 & 0.140 & 0.123 & 37.5 & 35.1 & 56.3 & 5.82 & 0.0462 & 2.10  \\
DimeNet++ + GeoTMI & $X, \tilde{X} / \tilde{X}$ &14.0 & 0.127 & 0.109 & 35.0 & 32.6 & 55.3 & 5.29 & 0.0418 & 1.90 \\
\midrule
\midrule
\multicolumn{2}{l}{Improvements in SchNet (\%)} & 0.00	& 33.2 &	18.1 &	15.3 & 16.5 & 13.3 & 22.6 & 4.87 & 7.39 \\
\multicolumn{2}{l}{Improvements in DimeNet++ (\%)} & 17.2 &	9.29 &  11.4 &	6.72 & 	7.12 &	1.78 & 9.11 &	9.52 & 9.52
 \\
\bottomrule
\end{tabular}
}
\vskip -0.1in
\end{table*}

\subsection{Performance of GeoTMI based on OOD data for QM9's properties} \label{A.4}
We identified the out-of-distribution (OOD) generalization ability of GeoTMI using the EGNN model on QM9's properties.
To this end, we arranged 100,000, 18,000, and 13,000 molecules for training, validation, and testing, respectively, based on a scaffold split, ensuring that the molecules in the OOD were included in the test set.
The split is based on the Bemis–Murcko scaffold \cite{Bemis} implemented in RDKit \cite{rdkit} library.
\Cref{tab:pearson_coeff_scaffold_split} shows that GeoTMI has improved the model prediction performance regardless of OOD data for the tested properties.
The results show the robustness of GeoTMI in terms of OOD generalization ability.

\begin{table}[h]
\caption{
The MAEs for $R^2$, $\mu$, and $U_0$ in the QM9$_\mathrm{M}$.
We verified the performance of GeoTMI on testing datasets using random and scaffold splits, respectively.
For the MAE of each property, the same units are used as in \Cref{tab:QM9_results}.
}
\label{tab:pearson_coeff_scaffold_split}
\vskip 0.15in
\centering
\begin{tabular}{llccc}
\toprule
Split & Approach & $\mu$ & $R^2$ & $U_0$ \\
\midrule
\multirow{2}{*}{Random} & EGNN  & 0.133 & 5.60 & 17.4 \\
& EGNN + GeoTMI & 0.100 & 4.08 & 14.5 \\
\midrule
\multirow{2}{*}{Scaffold} & EGNN  & 0.195 & 10.4 & 33.0 \\
& EGNN + GeoTMI & 0.149 & 7.60 & 23.4 \\
\bottomrule
\end{tabular}
\vskip -0.1in
\end{table}

\newpage

\section{Experimental details}
\label{Appendix.C}
\subsection{Parameter details}
\textbf{QM9$\mathrm{_M}$.}
\label{QM9_hyperparameter}
We used the reported hyperparameters optimized for QM9 from previous studies for EGNN \cite{EGNN}, SchNet \cite{Schnetpack}, DimeNet++ \cite{DimeNet++}, and Transformer-M \cite{transformer-M}, respectively.
The search space of $\lambda$ is specified in \Cref{tab:QM9_lambda_search_space}.

\begin{table}[h]
\scriptsize
\caption{The search space of $\lambda$ on $\mathrm{QM9_M}$ task.}
\label{tab:QM9_lambda_search_space}
\vskip 0.15in
\begin{center}
\begin{small}
\begin{tabular}{lccc}
\toprule
Target & EGNN & DimeNet++ & SchNet \\
\midrule
$U_0$    & [0.1, 0.5, 1.0, 10.0] & 0.1 & [0.1, 0.5, 1.0, 5.0, 10.0] \\
$\mathrm{\mu}$   & [0.1, 0.5, 1.0, 10.0] & 0.1 & [0.1, 0.5, 1.0, 5.0, 10.0] \\
$\alpha$ & 0.1 & 0.1 & [0.1, 0.5, 1.0, 5.0, 10.0] \\
$\epsilon_\mathrm{HOMO}$  & 0.1 & 0.1 & 0.1 \\
$\epsilon_\mathrm{LUMO}$  & 0.1 & 0.1 & 0.1 \\
$\mathrm{GAP}$   & 0.1 & 0.1 & 0.1 \\
$R^2$   & [0.1, 0.5, 1.0, 10.0] & 0.1 & [0.1, 0.5, 1.0, 5.0, 10.0] \\
$C_v$   & 0.1 & 0.1 & 0.1 \\
$\mathrm{ZPVE}$  & 0.1 & 0.1 & 0.1 \\
\bottomrule
\end{tabular}
\end{small}
\end{center}
\vskip -0.1in
\end{table}

\textbf{Reaction barrier prediction}.
\label{rxn_hyperparameter}
For both DimeReaction models with and without GeoTMI, we used the same hyperparameters as \citet{Spiekermann2022} except for the number of epochs and batch sizes.
\Cref{tab:rxn_hyperparameter} shows the values of the number of epochs and batch sizes used in this work.
The search space of $\lambda$ is specified in \Cref{tab:lambda_search_space}.

\begin{table}[h]
\caption{
The hyperparameters used for training DimeReaction.
}
\label{tab:rxn_hyperparameter}
\vskip 0.15in
\begin{center}
\begin{tabular}{lccr}
\toprule
Parameter & CCSD(T)-UNI & B97-D3 \\
\midrule
Epoch    & 200 & 200 \\
Batch size & 32 & 64 \\
Warm-up epochs & 3 & 3 \\
\bottomrule
\end{tabular}
\end{center}
\vskip -0.1in
\end{table}

\begin{table}[h!]
\caption{The search space of $\lambda$ for CCSD(T)-UNI and B97-D3 datasets.}
\label{tab:lambda_search_space}
\vskip 0.15in
\begin{center}
\begin{tabular}{cccr}
\toprule
Parameter & CCSD(T)-UNI & B97-D3 \\
\midrule
$\lambda$    & [0.005, 0.01, 0.05, 0.1, 0.5, 1.0] & [0.005, 0.01, 0.05, 0.1, 0.5, 1.0] \\
\bottomrule
\end{tabular}
\end{center}
\vskip -0.1in
\end{table}

\end{document}